\documentclass{article}

\usepackage{graphicx}
\usepackage{caption}
\usepackage{subcaption}
\usepackage{booktabs} 
\usepackage{array}
\usepackage{makecell}

\usepackage[accepted]{icml2023}

\usepackage{microtype}
\usepackage[export]{adjustbox}
\usepackage{multirow}
\usepackage{graphicx}
\usepackage{xcolor}
\usepackage{booktabs} 
\usepackage{url}            

\usepackage{amsfonts}       
\usepackage{nicefrac}       
\usepackage{microtype}      
\usepackage{bm}
\usepackage{diagbox}
\usepackage{color} 
\usepackage{amssymb}
\usepackage{amsmath}
\usepackage{amsfonts}
\usepackage{mathtools}
\usepackage{lineno}
\usepackage[toc,page,header]{appendix}
\usepackage{minitoc}

\usepackage{mathrsfs}
\usepackage{amsthm}
\usepackage{amsmath}
\usepackage{amssymb, bm}
\usepackage{comment}
\usepackage{multirow}
\usepackage{tabularx}
\usepackage{multicol}
\usepackage[normalem]{ulem}
\usepackage{todonotes}
\setuptodonotes{inline}

\makeatother
\usepackage{wrapfig}

\usepackage{multicol, multirow}
\usepackage[ruled,vlined,algo2e]{algorithm2e}
\usepackage{algorithmic}
\usepackage{algorithm}

\newtheorem{theorem}{Theorem}
\newtheorem{definition}{Definition}
\newtheorem{proposition}{Proposition}

\newtheorem{lemma}{Lemma}

\newcommand{\cA}{{\mathcal A}}

\newcommand{\cG}{{\mathcal G}}

\newcommand{\cP}{{\mathcal P}} 
 
\newcommand{\cR}{{\mathcal R}} 
\newcommand{\cS}{{\mathcal S}} 

\newcommand{\cM}{{\mathcal M}}

\newcommand{\cY}{{\mathcal Y}}
\newcommand{\cX}{{\mathcal X}}
\newcommand{\cZ}{{\mathcal Z}}

\newcommand{\cN}{\mathcal N}

\SetKwInput{KwInput}{Input}                
\SetKwInput{KwOutput}{Output}

\providecommand{\customgenericname}{}
\newcommand{\newcustomtheorem}[2]{%
  \newenvironment{#1}[1]
  {%
  \renewcommand\customgenericname{#2}%
  \renewcommand\theinnercustomgeneric{##1}%
  \innercustomgeneric
  }
  {\endinnercustomgeneric}
}
\newcustomtheorem{customtheorem}{Theorem}

\usepackage{cleveref}

\begin{document}
\twocolumn[
\icmltitle{MANSA: Learning Fast and Slow in Multi-Agent Systems}




\begin{icmlauthorlist}
\icmlauthor{David Mguni}{h}
\icmlauthor{Haojun Chen}{pk}
\icmlauthor{Taher Jafferjee}{h}
\icmlauthor{Jianhong Wang}{imp}
\icmlauthor{Longfei Yue}{pk}
\icmlauthor{Xidong Feng}{h,ucl}
\icmlauthor{Stephen McAleer}{ind}
\icmlauthor{Feifei Tong}{h}
\icmlauthor{Jun Wang}{ucl}
\icmlauthor{Yaodong Yang$^\dag$}{pk}
\end{icmlauthorlist}

\icmlaffiliation{h}{Huawei R\&D}

\icmlaffiliation{pk}{Institute for AI, Peking University}
\icmlaffiliation{imp}{University of Manchester}
\icmlaffiliation{ucl}{University College, London}
\icmlaffiliation{ind}{Independent Researcher}
\icmlcorrespondingauthor{}{davidmguni@hotmail.com}
\icmlcorrespondingauthor{}{j.wang@ucl.ac.uk}
\icmlcorrespondingauthor{}{yaodong.yang@pku.edu.cn}


\vskip 0.3in
]



\printAffiliationsAndNotice
\begin{abstract}
In multi-agent reinforcement learning (MARL), 
independent learning (IL) often shows remarkable performance 
and easily scales with the number of agents. Yet, using IL can be inefficient and runs the risk of failing to successfully train,  particularly in scenarios that require agents to coordinate their actions. 
Using 
centralised learning (CL) enables MARL agents to quickly learn how to coordinate their behaviour but employing CL everywhere is often prohibitively expensive in real-world applications. Besides, using CL in value-based methods often needs strong representational constraints (e.g. individual-global-max condition) that can lead to poor performance if violated.  
In this paper, we introduce a novel plug \& play IL framework named \textbf{M}ulti-\textbf{A}gent \textbf{N}etwork \textbf{S}election \textbf{A}lgorithm (MANSA) which selectively employs CL only at states that require coordination. At its core, MANSA has an additional agent that uses \textit{switching controls} to quickly learn the best states to activate CL during training,  
using CL only where necessary and vastly reducing the computational burden of CL. 
Our theory proves MANSA
preserves cooperative MARL convergence properties, boosts IL performance and can optimally make use of a fixed budget on the number CL calls. We show empirically in Level-based Foraging (LBF) and StarCraft Multi-agent Challenge (SMAC) that MANSA achieves fast, superior and more reliable  performance while making 40\% fewer CL calls in SMAC and using CL at only 1\% CL calls in LBF. 
\end{abstract}

\section{Introduction}













Multi-agent reinforcement learning (MARL) has emerged as a powerful framework that enables autonomous agents to complete various tasks in areas such as autonomous driving \citep{zhou2020smarts}, swarm robotics \citep{mguni2018decentralised,mguni2019coordinating} and smart grids \cite{wang2021multi,qiu2021multi,qiu2022mean}. Among MARL methods are a class of algorithms known as  independent learners (IL) e.g. independent Q learning \citep{tan1993multi}. IL decomposes a MARL problem with $N$ agents into $N$ decentralised single-agent problems. In this way, each agent treats other agents as part of the environment  which provides a straightforward way of training agents in a decentralised manner. Since the agents ignore other agents, IL can be trained quickly  as each agent's learning process is contingent on only its local observations and own actions. This is efficient in scenarios that require only weak interactions between agents \citep{kok2004sparse}.  

Despite these apparent benefits, training MARL using IL has several formidable drawbacks: with no ability to observe the actions of other agents, random occurrences of successful coordination among IL agents are improbable, causing IL methods to sometimes struggle in tasks that require coordination
\citep{hernandez2017survey}. Also, ignoring other agents' influence on the system means from the agent's perspective, the environment can appear non-stationary which precludes convergence guarantees \cite{yang2020overview}.  

 
%
%
%

On the other hand, MARL learners can be trained in simulated environments in which agents can be provided with other agents' observations and other state information. Centralised training and decentralised execution (CT-DE) \citep{kraemer2016multi, foerster2018counterfactual, mcaleer2022escher} is a framework that uses a centralised critic that exploits global information during training while performing execution in a decentralised fashion. With this added information during training, agents can learn to condition their policies on other agents' actions which mitigates the appearance of non-stationarity.  The CT-DE framework has become a central MARL paradigm and is the basis of popular methods such as QMIX \citep{rashid2018qmix}, SPOT-AC \citep{mguni2021learning} and COMA \citep{foerster2018counterfactual}. 
Various studies have conjectured that CT-DE can speed up training
by fostering cooperative behaviour and stabilising training. This is useful when 
there is a strong coordination component that produces a need for global observations during training \citep{sharma2021survey}. 
%
%
Nevertheless, CT-DE suffers from an explosive growth in complexity since the joint action-state space grows exponentially with the number of agents \citep{deng2023complexity}. Consequently, CT-DE methods can  require large numbers of samples to complete training. In regions in which the agents do not strongly interact, this added complexity can prove to be an unnecessary burden as agents do not benefit from global information \citep{kok2004sparse}. Fig. \ref{figure:traffic_junction} shows an example scenario in which the agents are required to coordinate only at a small subregion. 

To mitigate the explosive growth in complexity and enable CT-DE to scale, various CT-DE algorithms such as QMIX \citep{rashid2018qmix}, VDN \cite{sunehag2018value} decompose the joint value function into factors that depend only on individual agents. The representational constraints needed to achieve such decompositions can lead to provably poor exploration and suboptimality \citep{mahajan2019maven}. For example, QMIX requires a monotonicity constraint that can produce suboptimal value approximation. 
%
%
%
\begin{figure}[h]
    \centering
    \begin{subfigure}[b]{0.23\textwidth}
        
        \includegraphics[width=\textwidth, height=3cm]{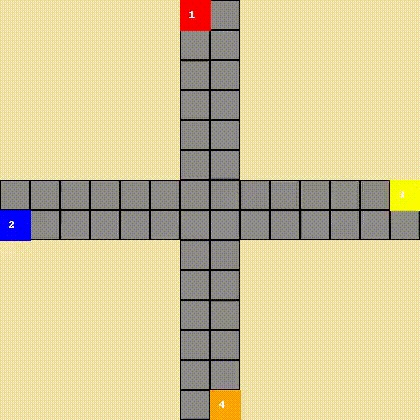}
    \end{subfigure}
    \begin{subfigure}[b]{0.23\textwidth}
        \includegraphics[width=\textwidth, height=3cm]{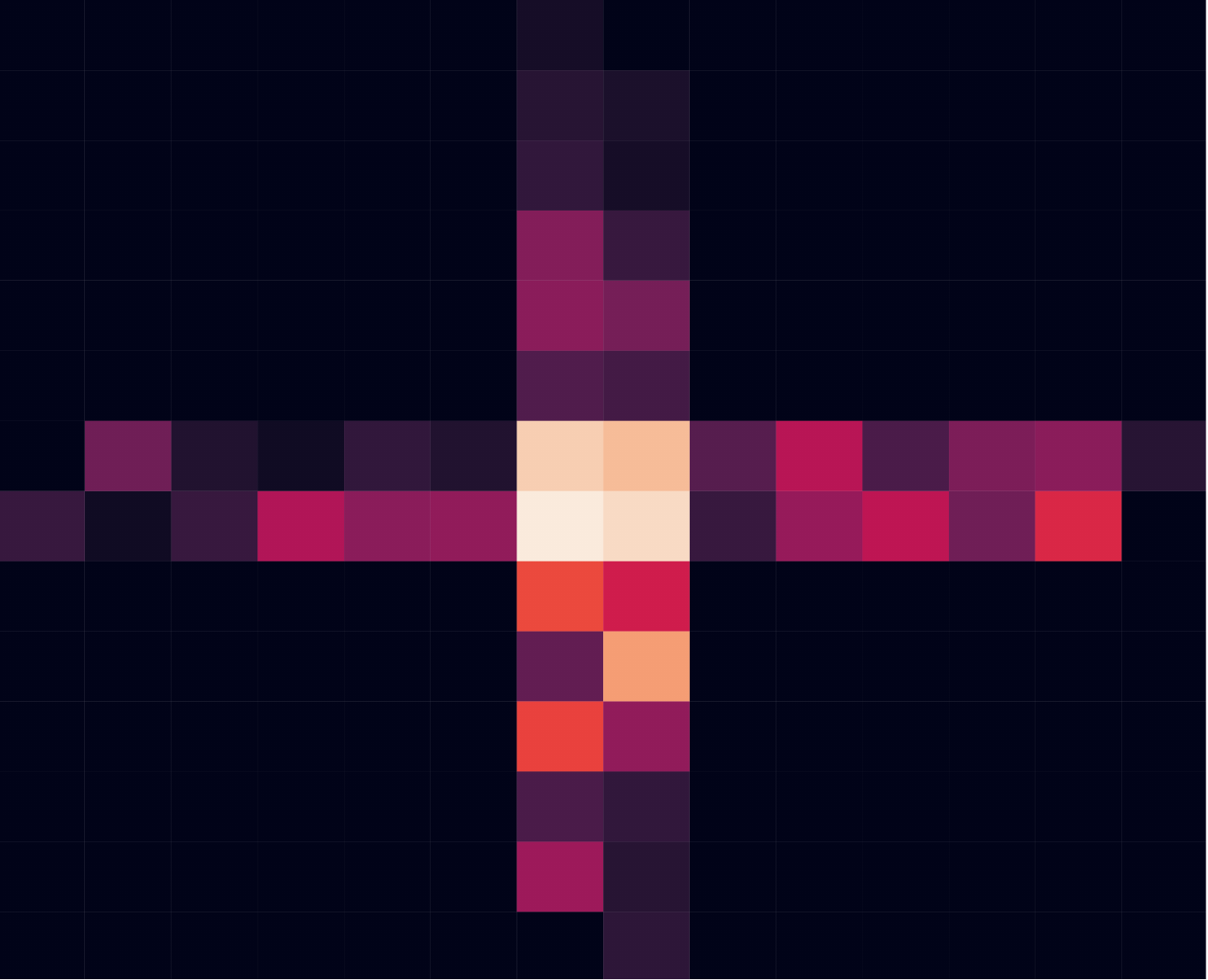}
    \end{subfigure}
    \caption{\textit{Left.} In this Traffic Junction scenario,  to avoid collisions agents (coloured squares) need only coordinate at the intersection. Before, their actions do not affect others so using IL at these states is sufficient. \emph{Right.} Heatmap of MANSA's CL calls. MANSA activates CL most at the intersection where coordination is needed.}
    \label{figure:traffic_junction} \vspace{-0.4 cm}
\end{figure}


To tackle these issues, we introduce a general plug \& play MARL framework, MANSA which optimally selects where in the environment to call on centralised learners to boost IL during training. MANSA involves  a decentralised learning method, 
a centralised critic network, and, an \textit{adaptive} reinforcement learning (RL) agent that presides over when CL or IL is used. 
Specifically, the additional agent determines at which states to activate CL 
while IL 
is used at all other states. This is in contrast to current MARL methods that use solely either CL or IL at all states throughout training. A key feature of MANSA is the novel combination of RL and a form of policy known as \textit{switching controls} \citep{mguni2022timing,mguniligs,mguni2021learning}. Switching controls are policies that introduce a switch mechanism that affects some control process in a dynamical system \cite{mguni2018viscosity}. 
In our case, as we show, this enables the adaptive RL agent to quickly determine where to switch to CL while the off-policy IL and off-policy CL jointly learn
from the gathered experience of whichever learner interacts with the environment (at any one time only one of CL or IL interacts with the environment) and minimise unnecessary CL calls during training.  This allows the benefits of both algorithm classes to be leveraged while overcoming some of the issues of any one class. Moreover, the binary decision space of switching controls means that the adaptive RL agent can rapidly determine the states where CL is beneficial while the MARL agents learn.


Since CL calls are expensive, it can be useful to consider enforcing a fixed budget on the number of CL calls during training. To this end in Sec. \ref{sec:mansa_budget}, we extend MANSA to enable it to solve MARL problems while respecting a budgetary constraint of the number of allowed CL calls during training.

Overall, MANSA has several advantages:  
%

    $\bullet$ By switching to CL only at the set of states in which it is beneficial while leveraging the benefits of IL, MANSA increases the learning efficiency of CT-DE (see Sec. 6.1).
    \newline
    $\bullet$ MANSA activates CL when (and only when) required resulting in MANSA boosting IL performance and enabling IL to tackle tasks which using IL would otherwise lead to coordination failure (see Sec. 6.2.2).
    \newline
    $\bullet$  MANSA minimises the number of times that CL is called (and hence the global information is used during training) while either matching or improving the performance of fully CL methods (see Sec. 6.2.1).    Additionally, MANSA allows for a fixed budget for calls of CL (see Sec. 6.3). 
    \newline
    $\bullet$ MANSA is a plug \& play framework which seamlessly adopts any MARL algorithm (see Sec. 6.2). 

%
To enable MANSA to perform successfully, we tackle several challenges.
First, including a new adaptive RL agent that learns while the $N$ MARL agents are training can occasion convergence issues. Second,
the adaptive RL agent uses switching controls which differs from the frameworks of standard RL. To this end, we prove MANSA preserves the MARL convergence properties (Theorem \ref{theorem:existence}) and boosts the performance of IL agents (Prop. \ref{NE_improve_prop}). We then characterise the optimal CL activation points with an online condition enabling it to quickly determine where switching to CL is beneficial during the agents' training phase (Prop. \ref{prop:switching_times}). 

When the problem includes budgetary constraints on the number of allowed CL calls, as the number of CL calls accumulates there is less freedom to execute more CL activations further on during training. Therefore to make optimal use of the allowed number CL calls, it is necessary to learn a policy that optimally decides whether to activate CL calls \textit{given its remaining budget}.  We resolve this by using a \textit{state augmentation} technique which treats the remaining budget as a state component (Theorem \ref{thm:optimal_policy_budget}). State augmentation techniques originated in control theory \cite{daryin2005nonlinear} and have recently been adapted to single agent RL \cite{sootla2022saute,mguni2022timing}.     

\section{Related Work}
A key aim of the CT-DE framework is to ensure the policies it generates are consistent with the desired system goal. One framework to fulfil this from a game theoretical perspective is called Markov Convex game (MCG) \cite{wang2020shapley, wang2022shaq}. A necessary condition for the MCG is the Individual-Global-Max (IGM) principle  \citep{son2019qtran}. To realise the IGM in the CT-DE framework, QMIX \cite{rashid2018qmix} and VDN \cite{sunehag2018value} propose two sufficient conditions of IGM to factorise the joint action-value function. Crucially, such decompositions and limited by the action-value function classes and the systems that do not adhere to these conditions \citep{wangqplex}. 

Several methods have been proposed to address this structural limitation. QPLEX \citep{wangqplex} uses a dueling network architecture to factor the joint action-value function avoiding representational restrictions. Nevertheless, QPLEX has been shown to fail in simple tasks with non-monotonic value functions \citep{rashid2020weighted}. 
QTRAN \citep{son2019qtran} formulates the MARL problem as a constrained optimisation problem with L2 penalties for decentralisation. Nevertheless, QTRAN has been shown to scale poorly in complex
MARL tasks such as SMAC \citep{peng2021facmac}. WQMIX \citep{rashid2020weighted} considers a weighted projection which is weighted towards better performing joint actions. At the core of these techniques are heuristics that do not guarantee IGM consistency. Consequently, achieving full expressiveness of the IGM function class with scalability remains an open challenge for MARL. 

Actor-critic methods such as COMA \citep{foerster2018counterfactual} and MADDPG \citep{lowe2017multi} are popular methods within MARL. These methods involve a centralised critic but nonetheless do not impose restrictions to represent the joint-action value function. Nevertheless, these methods are outperformed by value-based methods such as QMIX \cite{rashid2018qmix} and SHAQ \cite{wang2022shaq} on standard MARL benchmarks e.g. StarCraft Multi-Agent Challenge (SMAC) \citep{peng2021facmac}. MAPPO \citep{yu2022surprising} which is a leading actor-critic method with a centralised value function, extends a popular single-agent RL method, PPO \citep{schulman2017proximal} to MARL. Nevertheless, in some tasks,  MAPPO has been shown to be outperformed by IL, specifically, PPO \cite{schulman2017proximal} with only modest hyperparameter tuning \cite{de2020independent}. 
Consequently, in this paper, we realise our framework within value-based methods. Nevertheless, MANSA's plug \& play facility supports the extension to actor-critic methods.  

Several papers have explored the issue of exploiting localility of the agents' interactions
in different ways. Early works such as \citep{kok2004sparse} tackle the problem in learning in systems with sparse subregions. Such works make stringent assumptions that require the global coordination requirements of the system to be known beforehand. Moreover, other works centered on detecting where in the state space global or extra information is required to obtain a good policy. These works take the approach of detecting the influence of other agents on the reward signal. This approach is highly limited in our setting where the reward signal is allowed to be both a priori unknown and noisy.

\section{MANSA} 
%
A fully cooperative multi-agent system is modelled by a decentralised-Markov decision process (dec-MDP). A dec-MDP is an augmented MDP involving two or more agents $\{1,\ldots, N\}=:\mathcal{N}$ with a common goal that each independently decide actions to take which they do so simultaneously over many time steps. Formally, a dec-MDP is a tuple $\mathfrak{M}=\langle \mathcal{N},\mathcal{S},\left(\mathcal{A}_{i}\right)_{i\in\mathcal{N}},P,R,\gamma\rangle$ where $\mathcal{S}$ is the finite set of states, $\mathcal{A}_i$ is an action set for agent $i\in\mathcal{N}$. At each time $t\in 0,1,\ldots,$ the system is in state $s_t\in\mathcal{S}$ and each agent $i\in\mathcal{N}$ takes an action $a^i_t\in\mathcal{A}_i$. The \textit{joint action}\ $\boldsymbol{a}_t=(a^1_t,\ldots, a^N_t)\in\boldsymbol{\mathcal{A}}\equiv\times_{i=1}^N\cA_i$  produces an immediate reward $r\sim R(s_t,\boldsymbol{a}_t)$ where $R:\mathcal{S}\times\boldsymbol{\mathcal{A}}\to\mathcal{P}(D)$ is the team reward function that all agents jointly seek to maximise and
where $D$ is a compact subset of $\mathbb{R}$ and $\cP$ is some distribution on $\mathbb{R}$. Lastly, $P:\mathcal{S} \times \boldsymbol{\mathcal{A}} \times \mathcal{S} \rightarrow [0, 1]$ is the probability function describing the system dynamics.  We consider a partially observable system so that given the system is in the state $s_t\in \cS$, each agent $i\in \cN$ makes only local observations $\tau_{t,i}=O(s_t,i)$ where $O:\cS\times \cN \to \cZ_i$ is the observation function and $\cZ_i$ is the set of local observations for agent $i$.  To decide its action each agent samples its \textit{Markov policy} $\pi_{i,\theta_i}: \cZ_i \times \mathcal{A}_i \rightarrow [0,1]$ which is parameterised by the vector $\theta_i\in\mathbb{R}^d$ and is contained in $\Pi_i$. We occasionally drop the  parameter $\theta_i$ and write $\pi_{i}$ and we denote by $\boldsymbol{\Pi}:=\times_{i\in\mathcal{N}}\Pi_i$. For any agent and for any joint policy $\boldsymbol{\pi}\in \boldsymbol{\Pi}$, the state value and state-action value function are: $
v(s|{\boldsymbol{\pi}})=\mathbb{E}\left[\sum_{t=0}^\infty \gamma^tr\Big|s_0=s, \boldsymbol{a}\sim\boldsymbol{\pi}\right]$ and $
Q(s,\boldsymbol{a}|\boldsymbol{\pi})=\mathbb{E}\left[\sum_{t=0}^\infty \gamma^tr\Big|s_0=s,\boldsymbol{a_0}=\boldsymbol{a}; \boldsymbol{a}\sim\boldsymbol{\pi}\right]$ respectively. 

We now describe the core details of MANSA, how it learns to determine when to use a centralised learning process, and how it improves learning and performance. We then describe the agents' objectives and learning processes.

\subsection{Framework}

To tackle the challenges described, we equip each MARL agent with access to both a centralised learner, which we call {\fontfamily{cmss}\selectfont Central} and an independent learner, which we call {\fontfamily{cmss}\selectfont Independent}. MANSA includes an additional RL agent, {\fontfamily{cmss}\selectfont Global}, i.e., the switching controller, that decides on the states to activate {\fontfamily{cmss}\selectfont Central} during the agents' training phase while using {\fontfamily{cmss}\selectfont Independent} as the learning algorithm everywhere else. 

Fig. \ref{figure:mansa_schema} shows a schematic representation of MANSA. {\fontfamily{cmss}\selectfont Global} observes the global state $s_t$ of the environment and samples the discrete policy of the switching controller $g_t \sim \mathfrak{g}:\mathcal{S} \to \{0,1\}$. If $g_t = 0$, each of the $N$ agents in the environment use their respective local observations of the environment to generate actions $\boldsymbol{a}_t$ from the policy of {\fontfamily{cmss}\selectfont Independent}. If $g_t = 1$, $\boldsymbol{a}_t$ is generated from the policy of {\fontfamily{cmss}\selectfont Central} using the global state. The agents' actions $\boldsymbol{a}_t$ are executed in the environment and the loop repeats. The trajectories generated by this process are stored in a replay buffer from which {\fontfamily{cmss}\selectfont Global}, {\fontfamily{cmss}\selectfont Independent}, and {\fontfamily{cmss}\selectfont Central} are trained. MANSA includes an (additional) feature that imposes the condition that CL updates can only occur when the Global agent makes a CL call (i.e. when $g=1$). This feature serves as a useful tool when there is a need to ensure that communication costs are minimised during training while at the same time leveraging the benefits of both IL and CL.   


\begin{figure}[h!]
    \centering    
    \includegraphics[width=0.4\textwidth]{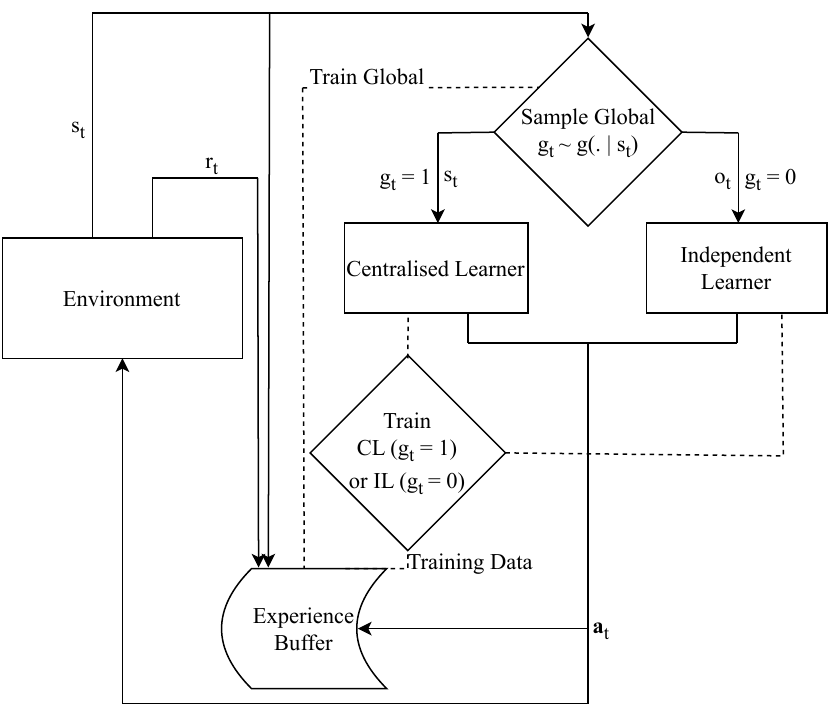}\vspace{-.4cm}
    \caption{\textcolor{black}{MANSA schematic.}}
    \label{figure:mansa_schema}
\end{figure}   

\begin{figure}[h!]
        \centering
        \includegraphics[width=0.45\textwidth]{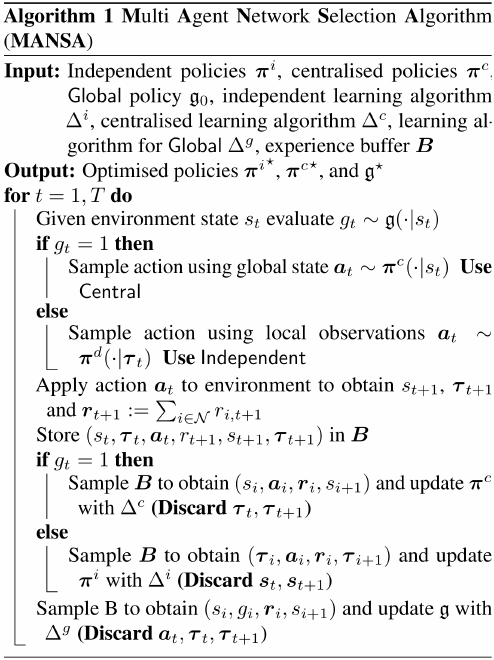}
    \label{fig:algo}
\end{figure}
The  {\fontfamily{cmss}\selectfont Global} agent is endowed with its own objective which captures its goal to improve the learning process and maximise the performance of the system of the $N$ MARL agents through its decisions of where to activate {\fontfamily{cmss}\selectfont Central}. 
To induce {\fontfamily{cmss}\selectfont Global} to selectively choose when to perform an activation, each activation incurs a fixed cost for {\fontfamily{cmss}\selectfont Global} which is quantified by a fixed constant $c>0$. These costs ensure that any activation of the CL critic must be beneficial to the performance of the system either at the current or subsequent states. The objective for {\fontfamily{cmss}\selectfont Global} is: 
\begin{align*}
\hspace{-3 mm}v_G(s|\boldsymbol{\pi},\mathfrak{g})  = \mathbb{E}_{g\sim\mathfrak{g}}\left[ \sum_{t=0}^\infty \gamma^t\left(r -c\cdot \boldsymbol{1}(g(s_t))\right)\Big|s_0=s; \boldsymbol{a}\sim\boldsymbol{\pi}\right],
\end{align*}
and {\fontfamily{cmss}\selectfont Global}'s action-value function is $Q_G(s,\boldsymbol{a}|\boldsymbol{\pi},\mathfrak{g})=\mathbb{E}_{g\sim\mathfrak{g}}\left[\sum_{t=0}^\infty \gamma^t(r-c\cdot \boldsymbol{1}(\mathfrak{g}(s_t)))|s_0=s,\boldsymbol{a_0}=\boldsymbol{a}; \boldsymbol{a}\sim\boldsymbol{\pi}\right]$.  

With this objective, {\fontfamily{cmss}\selectfont Global}'s goal is to maximise the system performance by activating {\fontfamily{cmss}\selectfont Central} at the required set of states to enable the agents to solve $\mathfrak{M}$ with the minimal number of CL activations. Therefore, by learning an optimal $\mathfrak{g}$, {\fontfamily{cmss}\selectfont Global} acquires the optimal policy for activating  {\fontfamily{cmss}\selectfont Central}.

Adding the agent {\fontfamily{cmss}\selectfont Global} with an objective distinct from the $N$ agents results in a non-cooperative Markov game $\cG=\langle \cN\times\{G\},\cS,\left((\cA_i)_{i\in\cN},\cA_G\right),P,(R,R_G),\gamma\rangle$ where $G$, $\cA_G:=\{0,1\}$ and $R_G(s,a,g):=R(s,a) -c\cdot \boldsymbol{1}(g)$ denote the {\fontfamily{cmss}\selectfont Global} agent, its action set and its reward function respectively. In MARL, having multiple learners with a payoff structure that is neither zero-sum nor a team game can occasion convergence issues \cite{shoham2008multiagent}. Moreover, unlike standard MARL frameworks, MANSA incorporates switching controls used by {\fontfamily{cmss}\selectfont Global}. Nevertheless in Sec. \ref{sec:convergence} we prove the convergence of MANSA under standard assumptions. 

%

\textbf{\textcolor{black}{Details on Architecture }}\newline
\textbf{MANSA's components.} 
We now describe a concrete realisation of MANSA's core components which consist of $N$ MARL agents, a CL RL algorithm as \textbf{{\fontfamily{cmss}\selectfont Central}}, an IL RL algorithm as  {\fontfamily{cmss}\selectfont Decentral} and a switching control RL algorithm as {\fontfamily{cmss}\selectfont Global}. Each (MA)RL component can be replaced by various other (MA)RL algorithms.

    $\bullet$ \textbf{$N$ MARL agents}. Each agent has two value-based policies. That is, each agent has (1) a policy induced by a value function that takes as input agent's \textit{global observation} which includes the joint action and global state, and (2) an action policy induced by a value function that takes as input only the agent's local observation.
    \newline
    $\bullet$ \textbf{Independent Q-Learning (IQL)}. In this paper, we use IQL \citep{tan1993multi} to train {\fontfamily{cmss}\selectfont Decentral}. IQL is a popular RL algorithm which is off-policy. 
    \newline
    $\bullet$ \textbf{QMIX}. For training {\fontfamily{cmss}\selectfont Central}, we use QMIX \citep{rashid2018qmix}, an off-policy MARL value-based method that accommodates only action value functions that adhere to a monotonicity constraint in the combination of the agents’ individual value functions.
    \newline
    $\bullet$ \textbf{Switching Control Policy} \cite{mguni2022timing}. A soft actor-critic (SAC) \cite{haarnoja2018soft} agent called {\fontfamily{cmss}\selectfont Global} whose policy's action set consists of $2$ actions: 1) use the centralised policy \textcolor{black}{(perform CL updates)}, 2) do not use the centralised policy \textcolor{black}{(perform IL updates)}. {\fontfamily{cmss}\selectfont Global} updates its policy $\mathfrak{g}$ while each agent learns their individual policy.

\textcolor{black}{The MANSA framework includes a feature that enables it to restrict the CL updates to only when {\fontfamily{cmss}\selectfont Global} executes a CL call (i.e. when $g_t=1$). In this way, communication occurs between the CL agents solely when {\fontfamily{cmss}\selectfont Global} performs a CL activation (no information is shared between IL and CL). This ensures the communication burden between agents is strictly limited during training.}


Note also the switching control mechanism results in a framework in which the problem facing {\fontfamily{cmss}\selectfont Global} has a markedly reduced computational complexity as compared with that facing the {\fontfamily{cmss}\selectfont Central} and {\fontfamily{cmss}\selectfont Decentral}  (though the learners share the same experiences). Crucially, the decision space for {\fontfamily{cmss}\selectfont Global} is $\mathcal{S}\times\{0,1\}$ i.e at each state it makes a binary decision. Consequently, the learning process for $\mathfrak{g}$ is much quicker than either {\fontfamily{cmss}\selectfont Central} or {\fontfamily{cmss}\selectfont Decentral}'s policy which must optimise over a decision space which is $|\mathcal{S}||\mathcal{A}|$ (choosing an action from its action space at every state) and $|\mathcal{S}||\mathcal{A}|^N$ respectively. This results in {\fontfamily{cmss}\selectfont Global} rapidly learning its optimal policy (relative to the base MARL learners).

\section{Convergence and Optimality of MANSA} \label{sec:convergence}

 \textcolor{black}{We now show that the MANSA framework, which induces an $N+1$ \textit{non-cooperative Markov game}, converges to the solution that both maximises the {\fontfamily{cmss}\selectfont Global} agent's value function and the agents' joint objective. With this, the {\fontfamily{cmss}\selectfont Global} agent learns to activate CL only at the set of states at which doing so improves the system performance of the MARL agents. The result is achieved through several steps: Theorem 1 shows MANSA learns the optimal solution for {\fontfamily{cmss}\selectfont Global} so that it activates CL only when it is  profitable to do so over the horizon of the problem (recall that each activation incurs a CL cost) while the agents' learn to maximise their objective. Prop. 1 proves the MANSA framework leads to higher system performance as compared to training the underlying base MARL method on its own. Finally, we characterise the optimal CL activation points and show that {\fontfamily{cmss}\selectfont Global} can use a condition on its action-value function that can be evaluated online to determine when to activate CL (for the case when {\fontfamily{cmss}\selectfont Global} uses a Q-learning variant).  All our results are built under Assumptions 1 - 7 (Sec. \ref{sec:notation_appendix} of the Appendix) which are standard in RL and stochastic approximation theory.
 }
 
 \textcolor{black}{The following theorem shows that for a fixed set of joint IL and CL policies, the solution of {\fontfamily{cmss}\selectfont Global}'s problem is a limit point of a sequence of Bellman operations acting on a value function (i). It then shows that the system in which both the IL, CL and {\fontfamily{cmss}\selectfont Global} agents train concurrently within the MANSA framework converges to the solution (ii). }





%

\begin{theorem}\label{theorem:existence}
\textbf{i)} Let $v_G:\mathcal{S}\to\mathbb{R}$ then for any fixed joint policies $\boldsymbol{\pi}^c,\boldsymbol{\pi}\in \boldsymbol{\Pi}$ the solution of {\fontfamily{cmss}\selectfont Global}'s problem is given by 
\begin{align}
\underset{k\to\infty}{\lim}T_G^kv_G(\cdot|\boldsymbol{\pi},\mathfrak{g})=\underset{\hat{\mathfrak{g}}}{\max}\;v_G(\cdot|\boldsymbol{\pi},\hat{\mathfrak{g}}),\end{align}
where $T_G$ is given by $
T_G v_G:=\max\Big\{\mathcal{M}^{\mathfrak{g},\boldsymbol{\pi}^c}Q_G,\underset{\boldsymbol{a}\in\boldsymbol{\mathcal{A}}}{\max}\;\left[ R_G+\gamma\sum_{s'\in\mathcal{S}}P(s';\cdot)v_G(s')\right]\Big\}$ and  
$
\mathcal{M}^{\mathfrak{g},\boldsymbol{\pi}^c}Q_G(s,\boldsymbol{a}|\cdot):=Q_G(s,\boldsymbol{\pi}^c(s)|\cdot)-c$ \textcolor{black}{which measures the expected return for {\fontfamily{cmss}\selectfont Global} following a switch to the CL joint policy minus the intervention cost $c$.}
\newline
\textbf{ii)} Given a system of convergent MARL learners of $\cM$, MANSA ensures the convergence of the system $\cG$ when {\fontfamily{cmss}\selectfont Global} uses a Q-learning variant.   

\end{theorem}
\textcolor{black}{Therefore, Theorem \ref{theorem:existence} proves the solution to {\fontfamily{cmss}\selectfont Global}'s problem in which {\fontfamily{cmss}\selectfont Global} optimally selects the set of states to activate CL can be obtained by computing the  limit of a (RL) dynamic programming procedure  (when {\fontfamily{cmss}\selectfont Global} uses a Q-learning variant). Secondly, it proves the MANSA system of $N+1$ agents jointly converges to the solution of $\cG$.} It is easy to see that an immediate consequence of the theorem is that MANSA learns to make the minimum number of CL calls required to learn the solution to the agents' joint problem since any additional CL calls would render the {\fontfamily{cmss}\selectfont Global} agent's policy suboptimal.
%

\textcolor{black}{Next we show MANSA improves performance outcomes:}

\begin{proposition}\label{NE_improve_prop}
There exists some finite integer $N$ such that $v(s|\boldsymbol{\tilde{\pi}}_m)\geq v(s|\boldsymbol{\pi}_m),\;\forall s \in\mathcal{S}$ for any $m\geq N$ where $\boldsymbol{\tilde{\pi}}_m$ and $\boldsymbol{\pi}_m$ are the joint policies after the $m^{th}$ learning iteration with and without {\fontfamily{cmss}\selectfont Global}'s influence respectively.
\end{proposition}
\textcolor{black}{The result shows that using the MANSA framework leads to improvements in the underlying MARL algorithm (as compared to training the MARL algorithm on its own). Note that \textit{a fortiori} Prop. \ref{NE_improve_prop} implies $v(s|\tilde{\boldsymbol{\pi}})\geq v(s|\boldsymbol{\pi}),\;\forall s \in\mathcal{S}$. }

The following result characterises {\fontfamily{cmss}\selectfont Global}'s policy $\mathfrak{g}$:

\begin{proposition}\label{prop:switching_times}
For any $s_t\in\mathcal{S}$ and for all $\boldsymbol{a}_t\in\boldsymbol{\mathcal{A}}$, the policy $\mathfrak{g}$ is given by: 
\begin{align}
\mathfrak{g}(\cdot|s_t)=\boldsymbol{1}_{\mathbb{R}_+}\left(\mathcal{M}^{\mathfrak{g},\boldsymbol{\pi}^c}Q_G(s_t,\boldsymbol{a}_t|\cdot)-\underset{\boldsymbol{a}_t\in\boldsymbol{\cA}}\max\; Q_G(s_t,\boldsymbol{a}_t|\boldsymbol{\pi},\mathfrak{g})\right),
\end{align} 
where $\boldsymbol{1}$ is the indicator function.
\end{proposition}
 Prop. \ref{prop:switching_times} provides characterisation of where {\fontfamily{cmss}\selectfont Global} should activate {\fontfamily{cmss}\selectfont Central}. The condition allows for the characterisation to be evaluated online during the learning phase.

\section{MANSA with a CL Call Budget} \label{sec:mansa_budget}

So far we have considered the case in which the aim is to solve the problem $\mathfrak{M}$ while using the minimum number of CL calls.  We now introduce a variant of MANSA, namely MANSA-B that aims to solve the problem while respecting a budgetary constraint of the number of allowed CL calls during training.  We show that by tracking its remaining budget the MANSA-B framework is able to learn a policy that makes optimal usage of its CL budget while respecting the budget constraint almost surely.  

The problem in which {\fontfamily{cmss}\selectfont Global} now faces a fixed budget on the number of CL calls  gives rise to the following constrained problem setting:
\begin{align*}
        \max\limits_{\mathfrak{g}}~& v_G(s|\boldsymbol{\pi},\mathfrak{g})\;\; 
        \text{s. t. } n - \sum_{k<\infty }\sum_{t_k\geq 0}\boldsymbol{1}(\mathfrak{g}(\cdot|s_{t_k}))\geq 0, \forall s\in\mathcal{S},   
\end{align*}
where $n\geq 0$ is a fixed value that represents the budget for the number CL activations and the index $k=1,\ldots$ represents the training episode count.  As in \citep{sootla2022saute,mguni2022timing}, we introduce a new variable $x_t$ that tracks the remaining number of activations: $x_t := n - \sum_{t\geq 0}\boldsymbol{1}(\mathfrak{g}(s_t))$ where the variable
$x_t$ is now treated as the new state variable which is a component in an augmented state space $\cX:=\cS\times\mathbb{N}$. We introduce the associated reward functions $\widetilde{R}:\cX\times\boldsymbol{\cA}\to\cP(D)$ and $\widetilde{R}_G:\cX\times\boldsymbol{\cA}\to\cP(D)$ and the probability transition  function $\widetilde{P}:\cX\times\boldsymbol{\cA}\times\cX\to[0,1]$ whose state space input is now replaced by $\cX$ and the {\fontfamily{cmss}\selectfont Global} value function for the game $\widetilde{\cG}=\langle \cN\times\{G\},\cS,\left((\cA_i)_{i\in\cN},\cA_G\right),\tilde{P},\tilde{R},\tilde{R}_G,\gamma\rangle$. We now prove MANSA-B ensures maximal performance for a given number of CL calls (CL call budget).

\begin{theorem} \label{thm:optimal_policy_budget} Consider 
the budgeted cooperative problem  $\widetilde \cG$, then
For any $\widetilde{v}:\cX\to \mathbb{R}$, the solution of $\widetilde{\cG}$ is given by $
\underset{k\to\infty}{\lim}\tilde{T}_G^k\widetilde{v}^{\boldsymbol{\pi}}=\underset{\mathfrak{g}}{\max}\;\widetilde v^{\boldsymbol{\pi},\mathfrak{g}}$, where {\fontfamily{cmss}\selectfont Global}'s optimal policy takes the Markovian form $\widetilde{\mathfrak{g}}(\cdot | \boldsymbol{x})$ for any $\boldsymbol{x}\equiv(x,s)\in\cX$. 
\end{theorem}

Theorem \ref{thm:optimal_policy_budget}  shows MANSA converges under standard assumptions to the solution of {\fontfamily{cmss}\selectfont Global}'s problem (and the dec-POMDP) when {\fontfamily{cmss}\selectfont Global} faces a CL call budget constraint.

\section{Experiments}\label{Section:Experiments}

%
We performed a series of experiments to test 
whether MANSA \textbf{1.} Enables MARL to solve multi-agent problems while reducing the number of CL calls \textbf{2.} Improves the performance of IL and reduces its failure modes \textbf{3.} Learns to optimise its use of CL under a CL call budget. We used the code accompanying the MARL benchmark study of \citet{papoudakis2021benchmarking} for the baselines.  For these experiments, we tested MANSA in Level-based Foraging (LBF) \citep{papoudakis2021benchmarking} and StarCraft Multi-Agent Challenge (SMAC) \citep{samvelyan2019starcraft}. These environments have specific features which in some cases are advantageous to CL, and in some cases to IL as well as a broad range of attributes as we describe below. We implemented MANSA on top QMIX \citep{rashid2018qmix} (as the CL) and IQL \citep{tan1993multi} (as the IL). We used SAC \citep{haarnoja2018soft} to learn the switching control policy itself. In all plots, dark lines represent averages over $3$ seeds and the shaded regions represent 95\% confidence intervals.

\textbf{Level-based Foraging (LBF).} In LBF an agent controls units of particular levels and there are apples of particular levels scattered around the map. Each agent's goal is to collect as much food as possible. Crucially, the agents can only collect a food if the cumulative level of the agents adjacent to the food that are executing the `collect' action is greater than or equal to the level of the food. As the agent and the food levels are randomly assigned, some food may be collectable by a single agent, while some food may require the coordination of all agents. 
LBF has the option of enforcing coordination (map names suffixed with "coop") by making the food level such that at least two agents are required to coordinate to collect any food. LBF tasks are  designed to sometimes require coordination to solve the problem, while other times needing little interaction between agents. 

\textbf{StarCraft Multi-Agent Challenge (SMAC).} The goal in SMAC  is for a team of units under an agent's control to defeat a team of units under an opponent's control.
Different maps in SMAC vary along several dimensions including heterogeneity of units, number of units, and terrain. These differences result in agents having to adopt varying degrees of coordination to solve different maps. For example, in \emph{so\_many\_baneling}, \textit{zealots} under the agent's control face a larger army of enemy \textit{banelings}. As banelings can cause significant `splash' damage, it is critical for units under the agent's control to cooperate and space out so as to minimise damage. Conversely, in \emph{corridor}, such cooperation may not be needed. Here, a small army of zealots under the agent's control face off against a large army of zerglings. The optimal strategy is for the zealots to wall-off a choke point and avoid getting surrounded. While it may seem that significant coordination is required to solve this map (i.e., all zealots converge to the choke point), in fact, it is not necessary. Due to location of the choke-point, the optimal actions for a zealot acting independently mirror those of a coordinated group -- IL is as good as CL in this case. Thus, the design of SMAC sometimes befits IL algorithms and sometimes CL algorithms. 

\subsection{Can MANSA learn to use CL less frequently in settings where CL is not required?}

For this experiment, we first studied MANSA in two normal-form (matrix) games: a \emph{coordination} game (specifically the Assurance Game) and the Non-Monotonic Team Game presented in \citet{rashid2018qmix}. We modified the reward function of the Assurance Game with a parameter $\alpha \in[0,1]$, as shown in Table \ref{table:normal_form_game_rewards}.
For $\alpha =0$, the reward function degenerates to the reward function of the standard Assurance game, while for $\alpha =1$, each agent gets a reward of $10$ irrespective of the other agent's action, that is, the game is completely decoupled. Similarly, in the non-monotonic team game, $\alpha$ parameterises the degree to which the reward structure of the game is non-monotonic. In this modified game, $\alpha=0$ represents a normal form game with a monotonic reward while $\alpha=1$ represents a non-monotonic reward function.
%
%
\begin{center}
\begin{table}[h!]
\begin{center}
\begin{tabular}{c|c|c} 
& Up & Down \\
\hline Up & $5(1+\alpha), 5(1+\alpha)$ & $10 \alpha, 10 \alpha$ \\
Down & $10 \alpha, 10 \alpha$ & \textcolor{black}{$10,10$}
\end{tabular}
\end{center}
\begin{center}
\begin{tabular}{c|c|c} 
& A & B \\
\hline A & $2 \alpha, 2 \alpha$ & 1,1 \\
\hline B & 1,1 & 8,8
\end{tabular}
\end{center}
\caption{Modified reward functions of Assurance Game (\emph{top}), and non-monotonic team game (\emph{bottom}).}
\label{table:normal_form_game_rewards}
\end{table}
\end{center}
%
%
%
\begin{figure}[t]
    \centering
    \begin{subfigure}[b]{.235\textwidth}
        \centering
        \includegraphics[width=\textwidth]{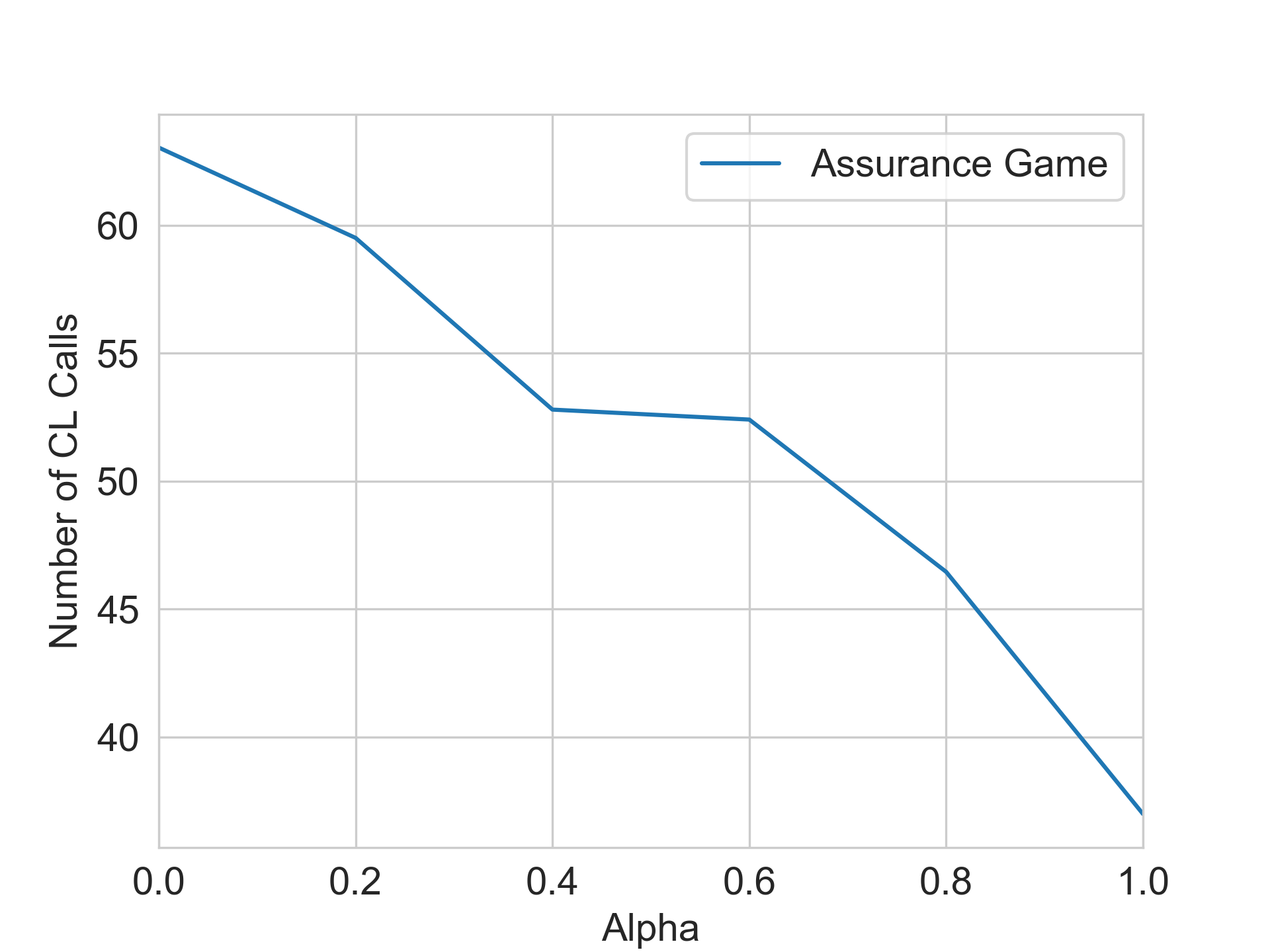}
    \end{subfigure}
    \begin{subfigure}[b]{.235\textwidth}
        \centering
        \includegraphics[width=\textwidth]{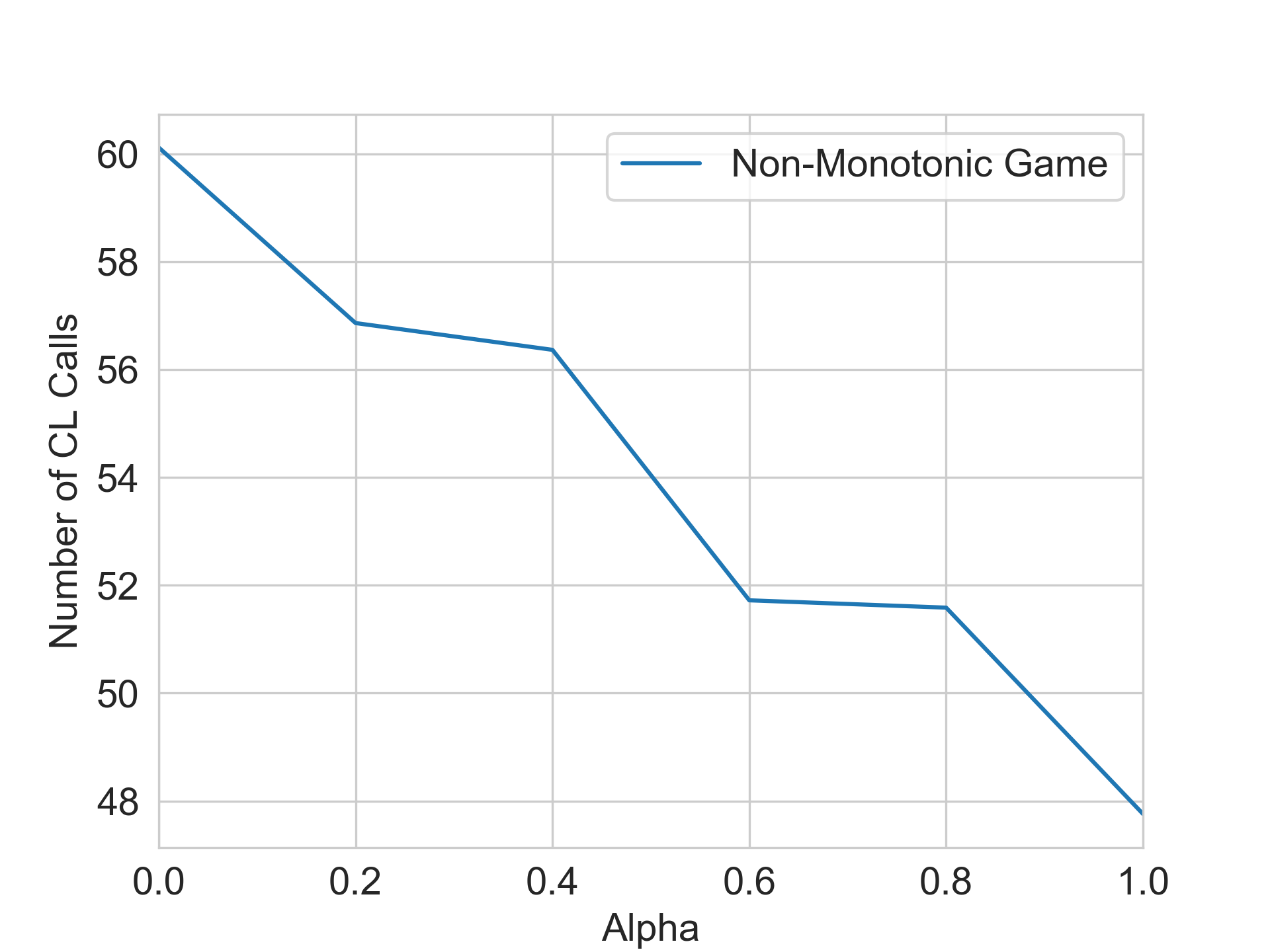}
    \end{subfigure}
        \begin{subfigure}[b]{.235\textwidth}
       \hspace{-2cm} \includegraphics[width=8.5cm, height = 3.5cm, left]{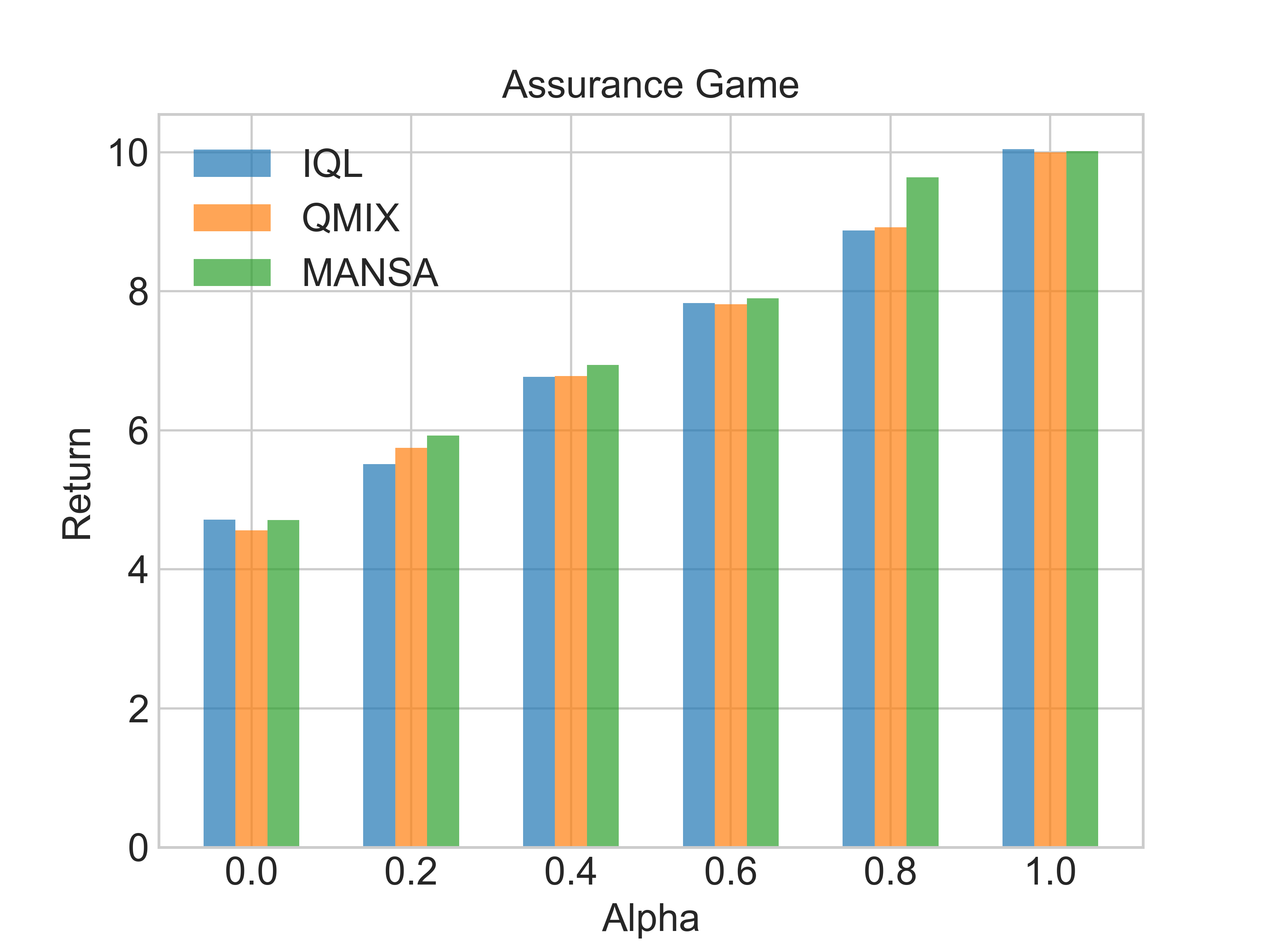}
    \end{subfigure}
\caption{\textbf{Normal form games.}  Total number of CL calls by MANSA in the Assurance Game (top left) and the non-monotonic team game (top right) \textcolor{black}{and end-of-training returns for MANSA, QMIX and IQL for various values of $\alpha$ (bottom)}. As the rewards in Assurance Game become more decoupled ($\alpha \to 1$) so the requirement for coordination becomes weaker, MANSA reduces the number of CL calls it makes during training.  In the Non-Monotonic game, as the extent of the monotonicity in the reward decreases ($\alpha \to $1), MANSA similarly reduces the number of CL (QMIX) calls. Note, in both cases MANSA makes a small number of calls to CL as {\fontfamily{cmss}\selectfont Global} initially explores both CL and IL. \textcolor{black}{Despite MANSA reducing its dependence on CL as $\alpha\to 1$, it achieves returns that are better or the same as the baselines for all $\alpha$.}}
\label{figure:normal_form_games}
\end{figure}
Fig. \ref{figure:normal_form_games} shows plots of $\alpha$ versus the number of calls to CL. In both games, higher values of $\alpha$ ought to result in less usage of CL, and as expected, as $\alpha$ increases, calls to the CL decrease and MANSA shows greater dependence on IL for training. This suggests MANSA is capable of selectively using CL with a high degree of granularity. It also provides strong evidence MANSA exercises thriftiness in its usage of CL in environments with no strong coordination aspect.


We next investigated MANSA's ability modulate its use of CL in LBF Foraging-8x8-2p-2f-coop-v1, a dynamic setting with many states and agents. To do this, we isolated three configurations of the LBF task that have strongly, medium and weakly coupled reward functions i.e. for the agents to solve the task, each case requires a specific level of coordination by the agents. The weakest case is a setting in which each food item can be collected by just one agent; in the medium level,  collecting each food item requires two agents to coordinate while in the strongest level, collecting each food item needs all agents to coordinate. For each case we measured the total number of CL calls made by MANSA over the course of training. As shown in Fig. \ref{fig:lbf_level_calls}, as the level of required coordination increases (from weak to strong), MANSA increases the number of CL calls to promote learning policies capable of coordination among the agents during their training phase.

Lastly, to confirm the usefulness of MANSA's switching control component, Section \ref{sec:ablation_switching} of the Appendix gives an ablation study in which we replaced the {\fontfamily{cmss}\selectfont Global} agent with a random switching controller.  compared to simply activating Central at random (line labelled ”random policy”). As is shown, removing MANSA's switching control aspect  leads to significant degredation in overall performance as compared with MANSA with its adaptive RL agent {\fontfamily{cmss}\selectfont Global}.

\begin{figure}[t]
        \centering
\includegraphics[width=0.3\textwidth]{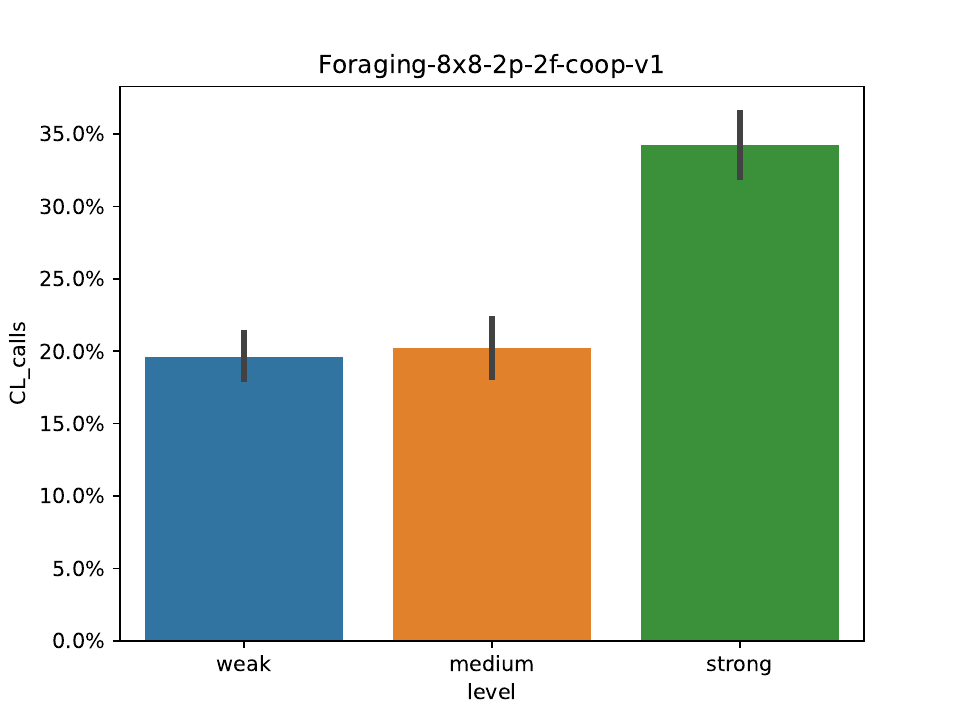}
       \caption{\textcolor{black}{Number of CL calls within LBF maps with varying degrees of coupling within the reward functions.}}\vspace{-4 mm}
    \label{fig:lbf_level_calls}
\end{figure}
\begin{figure}[b]\vspace{-0.6cm}
        \centering
\includegraphics[width=0.3\textwidth]{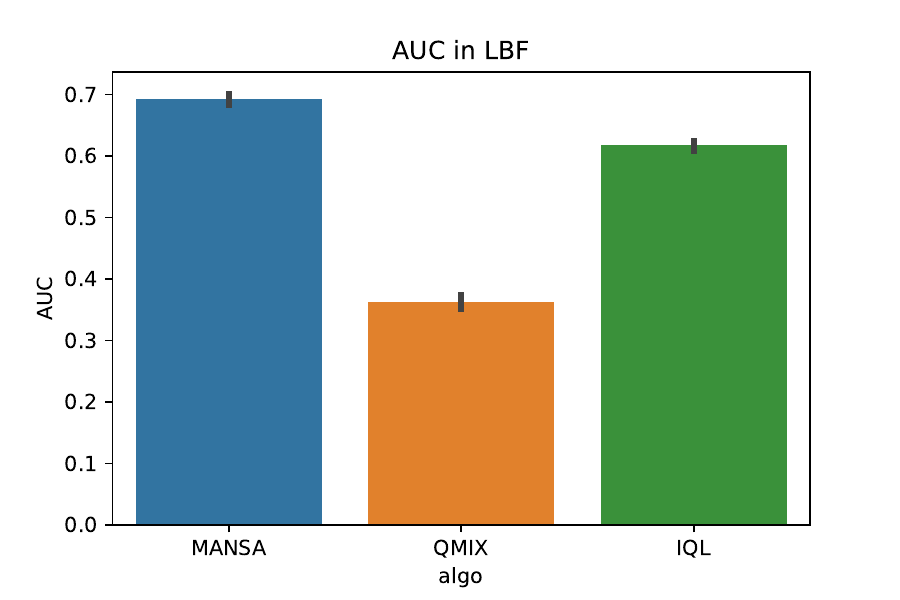}\vspace{-0.4 cm}
        \caption{\textcolor{black}{Aggregate (normalised) area under the curve (AUC) results across 10 LBF tasks. MANSA has superior aggregate performance, markedly outperforming the CL method (QMIX) and either matching or outperforming the IL method (IQL) on all tasks.}} \vspace{-2 mm}
    \label{fig:lbf}
\end{figure}
\begin{figure}[h]
    \begin{subfigure}{.235\textwidth}
        \includegraphics[width=\textwidth]{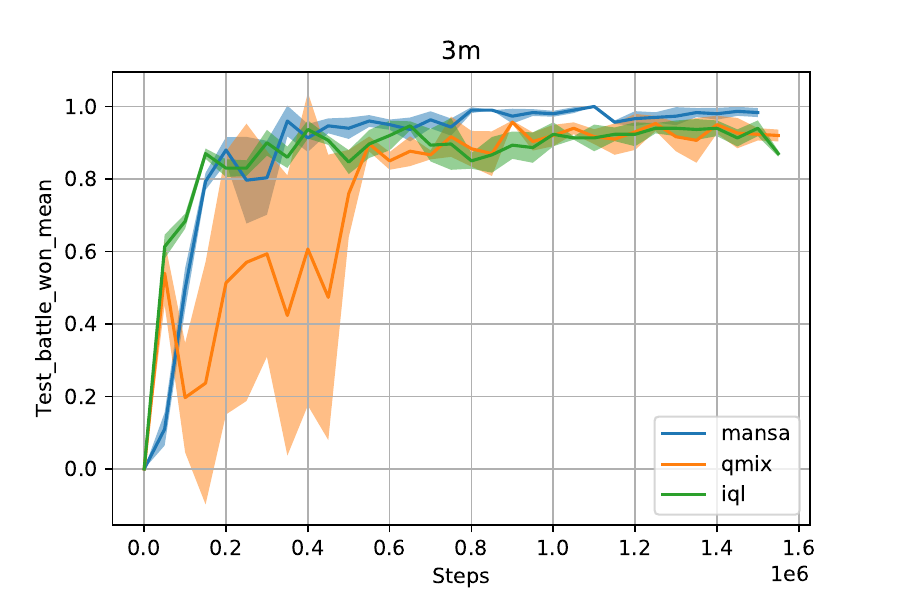}
        \includegraphics[width=\textwidth]{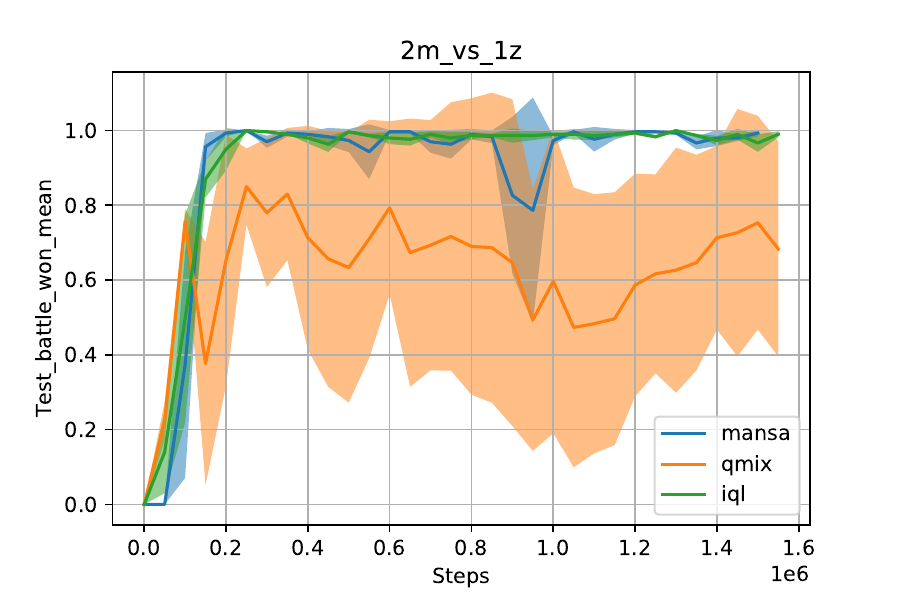}
        \end{subfigure}
    \begin{subfigure}{.235\textwidth}
        \includegraphics[width=\textwidth]{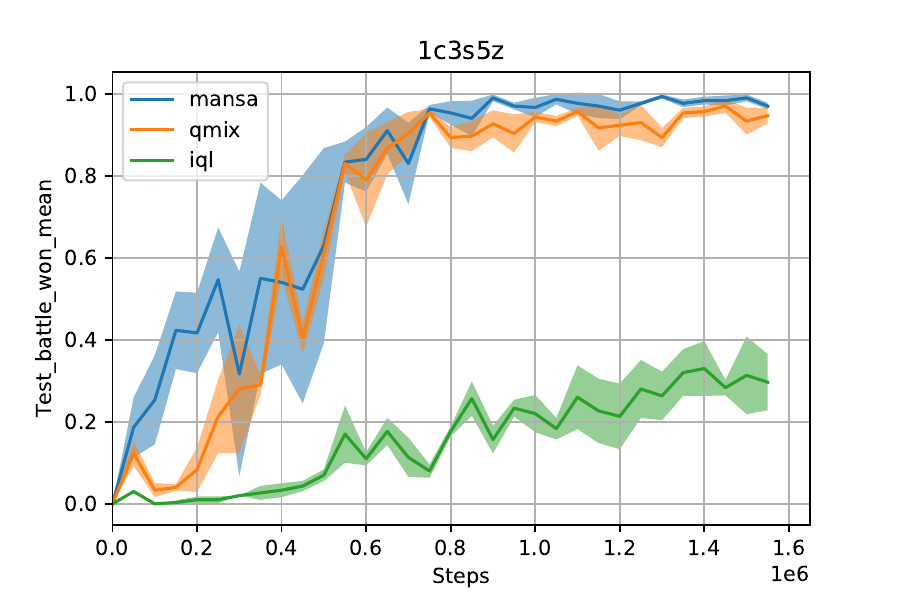}
        \includegraphics[width=\textwidth]{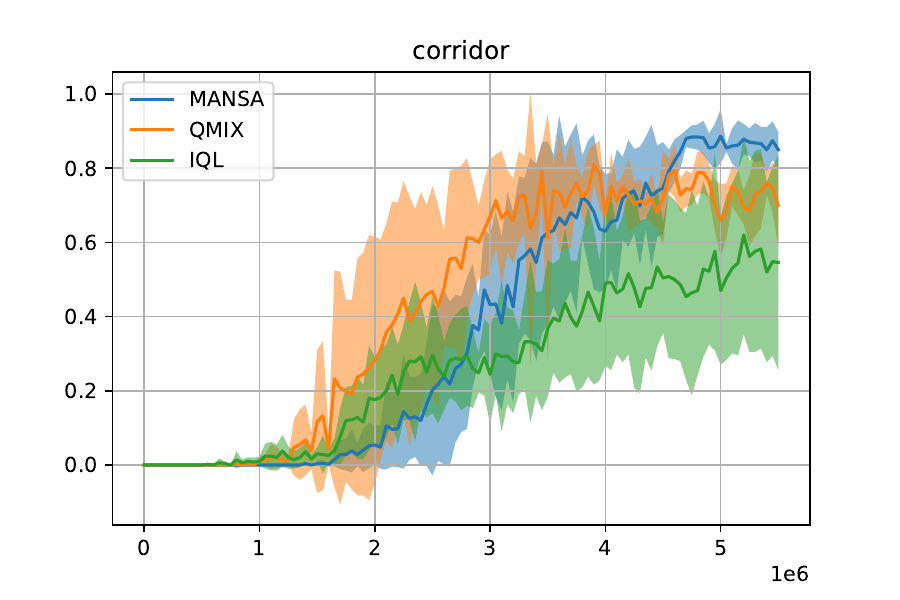}
    \end{subfigure}
\caption{\textcolor{black}{Learning curves in some individual SMAC maps. While QMIX fails to learn effective policies on all maps, and IQL on two maps, MANSA achieves high performance across the tasks.}}
\end{figure}

\subsection{Can MANSA improve the overall performance of IL and reduce failure modes?} 

We first examined this claim in LBF; Fig. \ref{fig:lbf} shows aggregated (normalised) area under the curve AUC performance curves of the tested algorithms (for individual plots see Sec. \ref{sec:performances_appendix} in the Appendix). MANSA outperforms both IQL and QMIX by a notable margin in half the maps (4 of 8). Moreover, even in maps where QMIX performs poorly, e.g., Foraging-10x10-3p-5f-v2, Foraging-10x10-5p-3f-v2, MANSA is able to use QMIX to significantly outperform IQL (compare performance of vanilla QMIX versus MANSA in plots in Sec. \ref{sec:performances_appendix}). This is due to  MANSA  correctly identifying states that benefit from CL (and those that do not) and there activating CL to achieve significant performance gains.
%
The empirical results serve to validate MANSA's preservation of MARL convergence properties and its ability to leverage both CL and IL to deliver higher performance. 
In Sec. \ref{sec:sc_ablation} of the Appendix, we show the results of an ablation study of the switching cost parameter. So long as the value of this hyper-parameter is roughly in the correct order of magnitude, MANSA performs well and thus is easy to tune.

\begin{figure}[t]
        \centering
        \includegraphics[width=0.35\textwidth] {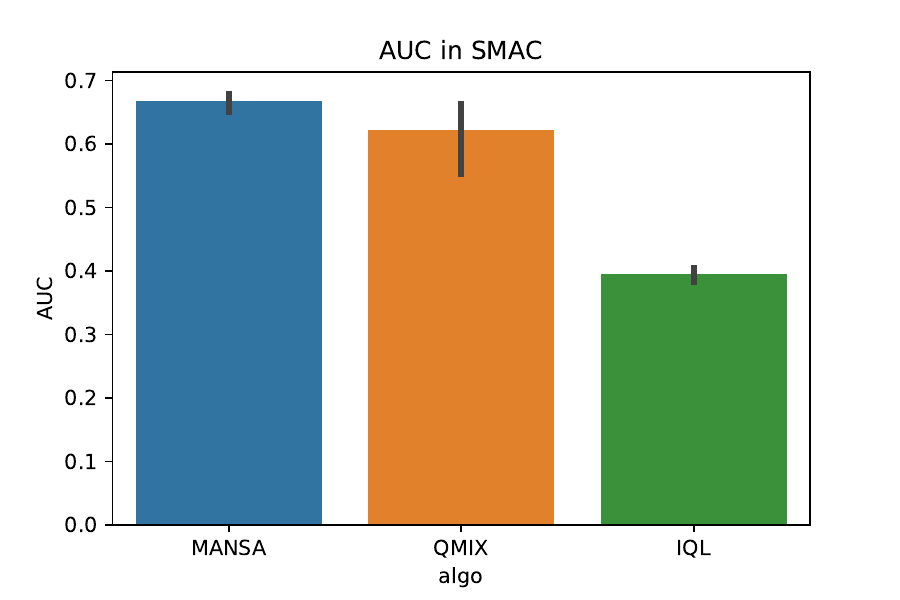} 
    \caption{\textcolor{black}{Aggregate normalised AUC  results across 9 SMAC maps. MANSA has superior performance and is not susceptible to learning failures unlike the base CL (QMIX) and IL methods (IQL).}}
\label{fig:smac} 
\end{figure}
%
%
\textcolor{black}{We next examined the claim in SMAC. Fig. \ref{fig:smac} shows the aggregated normalised AUC results across  a range of SMAC maps (for full set of plots of individual maps see Sec. \ref{sec:performances_appendix} in the Appendix).} MANSA's aggregate AUC performance is superior to both baselines. It also outperforms all baselines in all maps except \emph{3s5z\_vs\_3s6z}. MANSA's flexibile choice of MARL method allows it to avoid the failures of IQL in maps such as \emph{1c3s5z}, \emph{3s5z}, \emph{2s3z}, and \emph{MMM2} without heavily relying on CL (MANSA's CL call rates are shown in Table \ref{table:calls_to_cl_SMAC} of the Appendix). Similarly, MANSA avoids the failures of QMIX in \emph{2m\_vs\_1z} and \emph{corridor}. 

To validate the claim MANSA can reduce failure rates, we plotted the failure rates  of each algorithm (i.e. on how many tasks each algorithm failed by the total number of tasks) in Fig. \ref{figure:failure_rates}. We define a failure as achieving an end-of-training
win rate of less than $0.8$ on SMAC. IL and CL failed in 44\% (4 of 9) and 22\% (2 of 9), respectively, of the SMAC maps. 
%
\begin{figure}[h!]
        \centering
        \includegraphics[width=0.3\textwidth]{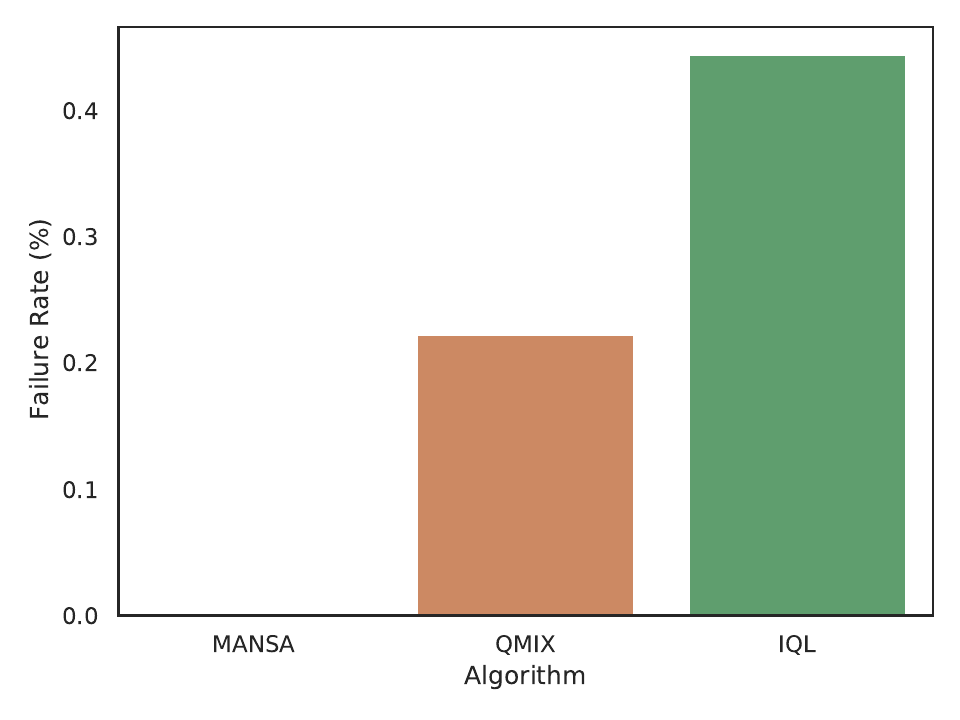}\vspace{-5 mm}
        \caption{Failure rates (number of failed tasks/total number of tasks) of each algorithm across all SMAC maps.
        }
        \label{figure:failure_rates}
\end{figure}

\subsection{MANSA is a Plug \& Play IL Enhancement Framework.}\label{sec:plugnplay}

To validate our claim that MANSA easily adopts MARL algorithms, we ran experiments with a stronger CL baseline to test if MANSA is still beneficial when the CL baseline is stronger than the IL baseline. To test this, we replaced QMIX in MANSA with a stronger CL component, W-QMIX. Fig. \ref{figure:further_lbf_exps} shows learning curves where, unlike in Figure \ref{fig:lbf}, IQL is outperformed by a CL algorithm, W-QMIX. For MANSA to achieve reasonable performance here, the switching controller ought to opt to use CL more frequently than IL even if this incurs a switching cost. Indeed, we see that in all maps, MANSA significantly outperforms the baselines, and from Table \ref{table:calls_to_cl_2} (see Sec. \ref{sec:appendix_plugplay} in the Appendix) we see that MANSA uses CL much more in these maps than the maps indicated in Table \ref{table:calls_to_cl_LBF}. Moreover, as with previous experiments, MANSA seems to have correctly identified states that benefit from CL (and those that do not) and have only used CL to achieve significant performance gains.
%
%
%
%

\begin{figure}[h]
    \centering
    \begin{subfigure}{.238\textwidth}
        \centering
        \includegraphics[width=\textwidth]{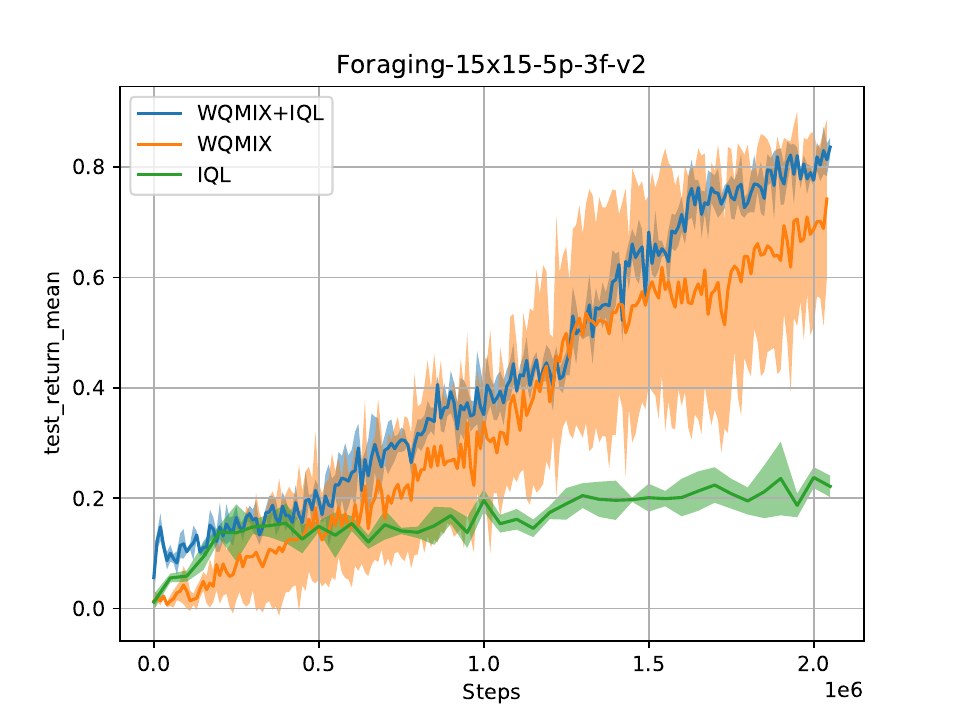}
        \caption{}
    \end{subfigure}
    \begin{subfigure}{.238\textwidth}
        \centering
        \includegraphics[width=\textwidth]{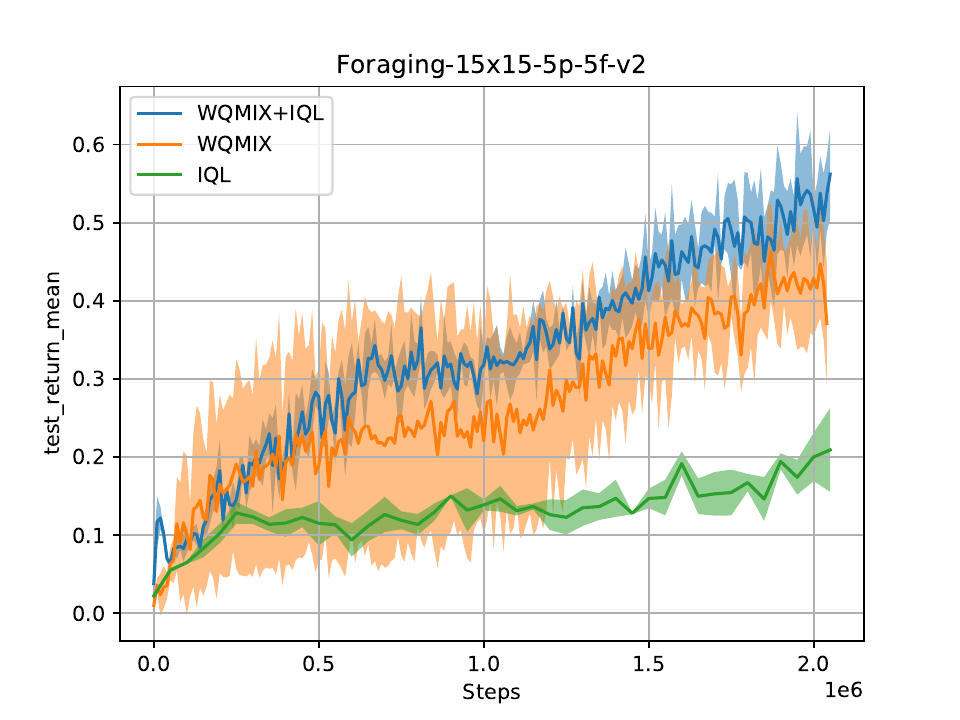}
        \caption{}
    \end{subfigure}
    \begin{subfigure}{.238\textwidth}
        \centering
        \includegraphics[width=\textwidth]{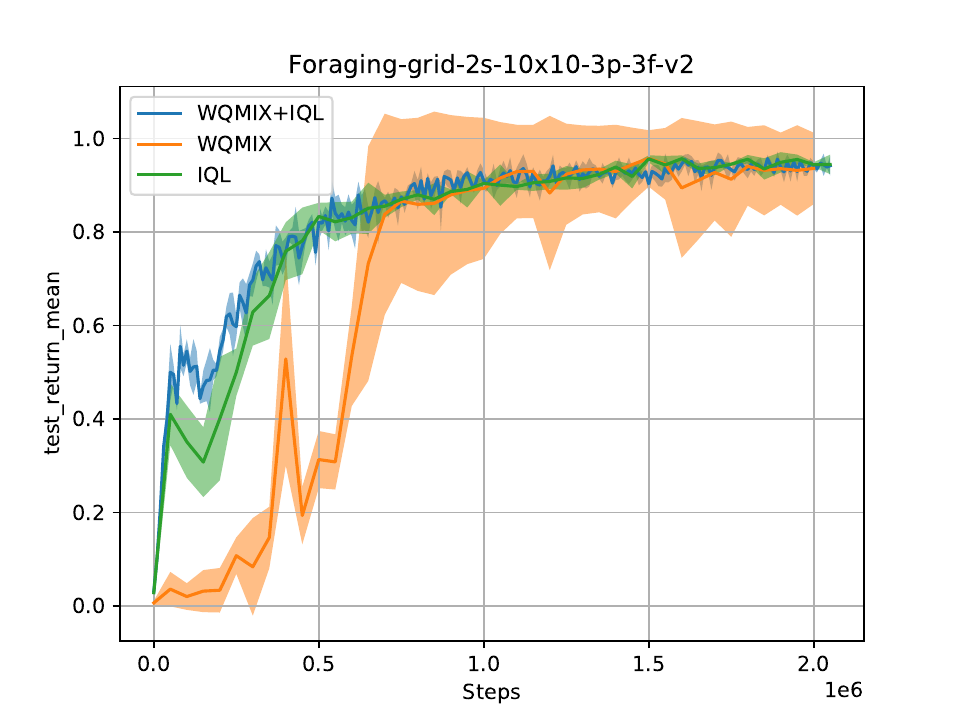}
        \caption{}
    \end{subfigure}
        \begin{subfigure}{.238\textwidth}
        \centering
        \includegraphics[width=\textwidth]{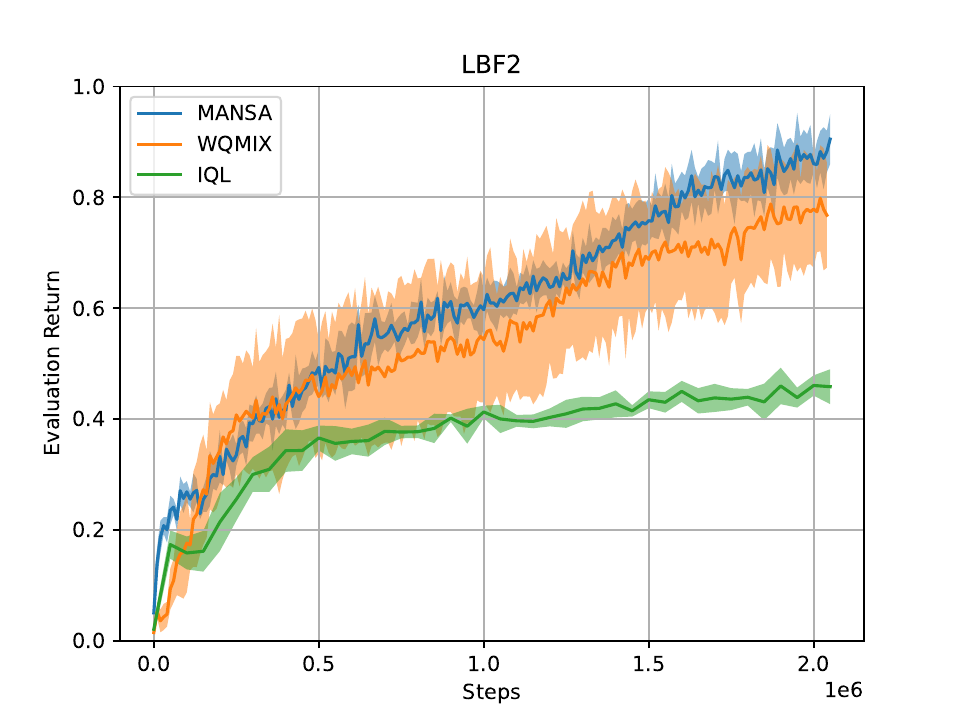}
        \caption{}
    \end{subfigure}
\caption{Learning curves on LBF with a stronger CL algorithm, W-QMIX. W-QMIX outperforms IQL on the selected maps, but MANSA is still able to leverage the advantages of both IL and CL to outperform the baselines.}
\label{figure:further_lbf_exps}
\end{figure}
%
%
%

\subsection{Can MANSA optimise use of CL under a budget?}

To validate our claim that MANSA-B optimises CL calls under a fixed budget, we ran MANSA-B in 4 SMAC maps with a varying CL call budget. Table \ref{table:budget_mansa} in the Appendix shows the Win rates comparing MANSA-B with various CL call budgets against MANSA (original). In 2 out of the 4 maps, MANSA achieves win rates of above 98\% despite a cap of 10\% on the original CL calls. As the budget increases to 50\%, MANSA achieves above 65\% win rates on all maps.

\subsection{Importance of Switching Controls}\label{sec:ablation_switching}
A key component of MANSA is the switching control mechanism. This enables the {\fontfamily{cmss}\selectfont Global} agent to select the states in which activating {\fontfamily{cmss}\selectfont Central} leads to performance improvements. To evaluate the impact of the switching control component, we compared the performance of MANSA with a version of MANSA which has the switching control replaced with an equal-chances Bernoulli Random Variable (i.e., at any given state, the {\fontfamily{cmss}\selectfont Global} decides whether or not to activate {\fontfamily{cmss}\selectfont Central} with equal probability) (note that always activating {\fontfamily{cmss}\selectfont Central} degenerates to QMIX and similarly, never activating {\fontfamily{cmss}\selectfont Central} degenerates to IQL). Figure \ref{fig:ablate_random} shows the comparison of the performances of the variants.  We examined the performance of the variants of MANSA in LBF Foraging-15x15-5p-3f-v2. As can be seen in the plot, incorporating the ability to learn an optimal switching control in MANSA (labelled "MANSA (OW-QMIX+IQL") leads to much better overall performance compared to simply activating {\fontfamily{cmss}\selectfont Central} at random (line labelled "random\_policy").
\begin{figure}[h!]\vspace{-0.7 cm}
      \centering \includegraphics[width=4cm, height=3.7cm]{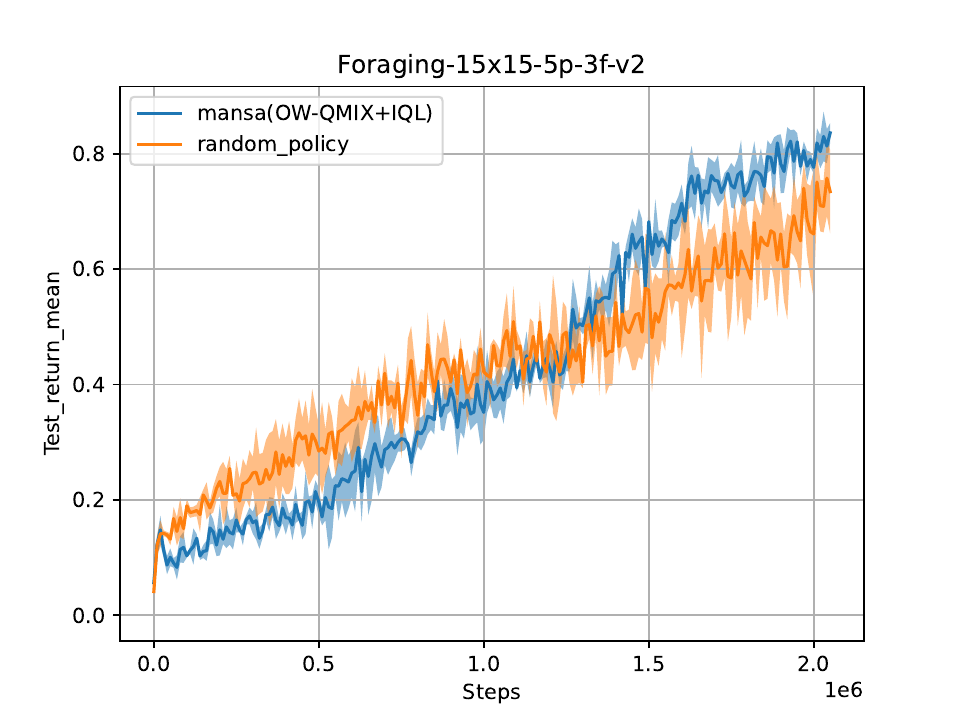}
\includegraphics[width=4cm, height=3.7cm]{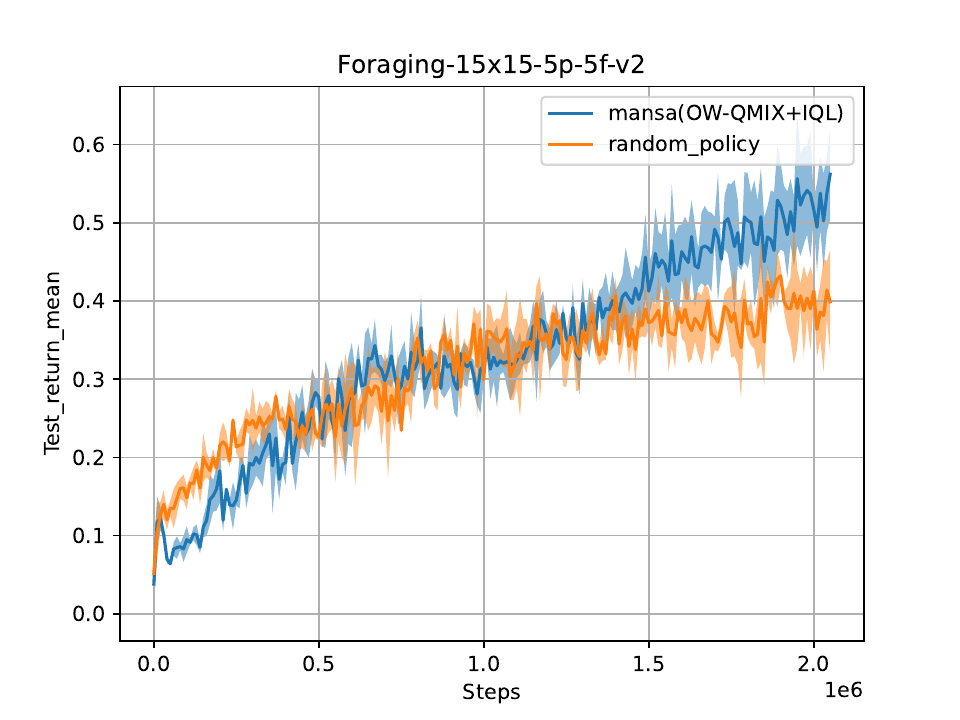}
\caption{MANSA's switching control component produces much better overall performance 
 ("MANSA (OW-QMIX+IQL") compared to randomly activating {\fontfamily{cmss}\selectfont Central} ("random\_policy").}
    \label{fig:ablate_random}\vspace{-0.7 cm}
\end{figure}
\section{Conclusion}
In this paper, we presented MANSA, a novel MARL framework for enhancing performance of IL training under a limited number of CL calls. MANSA combines IL and CL in a way that enables IL to leverage the benefits of CL while minimising the complexity burden and limitations of representational constraints suffered by CL methods. Conversely, MANSA mitigates the issues suffered by IL such as its inability to efficiently solve some coordination tasks and lack of convergence guarantees. In so doing, MANSA provides a framework that leverages each algorithm class and removes the split between IL and CL MARL training methods. Our theory proves MANSA preserves MARL convergence guarantees and improves MARL outcomes. Our empirical analyses present a detailed suite of experimental including LBF and SMAC --- in all these domains, MANSA improves performance,  reduces failure modes all the meanwhile minimising its use of CL. 
In future, we will consider the natural extension of the framework to encompass switching between various CL methods to leverage the benefits of their various factorisations. 

\section*{Acknowledgements}
 Yaodong Yang is supported in part by the National Key R\&D Program of China (2022D0114900) and CAAI-Huawei Mindspore Open Fund. Jianhong Wang is supported by UKRI Turing AI World-Leading Researcher Fellowship, EP/W002973/1.

\bibliography{main}
\bibliographystyle{icml2022}

\clearpage

\onecolumn
\part{{\Large{Appendix}}} 
\addcontentsline{toc}{section}{Appendix} 

\section{Assurance Game Construction}\label{sec:nfg}


In Sec. 6.1 we introduce a variant of the classic coordination game, the Assurance game. The matrix game presented in Sec. 6.1 is the supposition of the Assurance Game with a non-strategic component resulting in a new game whose entries are the supposition of the above components. The composition of the new game is calibrated by a parameter $\alpha$ which runs from $0$ to $1$. At its extreme points $0$ and $1$, the game degenerates into the Assurance game and the entirely non-strategic game. We begin by stating the reward function $R_i(a_i,a_j):\cA_i\times\cA_j\to\mathbb{R}$ for the agents $i,j\in\{1,2\}$.

\begin{align}
    R_i(a_i,a_j)=\alpha(\cR_i(a_i)+ \cR_j(a_j))+(1-\alpha)\mathfrak{R}_i(a_i,a_j), \;\;\; i,j\in\{1,2\}, \label{reward_interaction_ablation}
\end{align}
where $R_i:\cA_i\to\mathbb{R}$ and $R_j:\cA_j\to\mathbb{R}$ are bounded, real valued functions and $\cA_i$ and $\cA_j$ are compact sets.

We assume $\mathfrak{R}_i$ in \eqref{reward_interaction_ablation} can't be decoupled into a function of the form $\mathfrak{R}_i(a_i,a_j)=f(a_i)+ g(a_j)$. From \eqref{reward_interaction_ablation}, we see that when $\alpha = 0$, $R_i(a_i,a_j)=\mathfrak{R}_i(a_i,a_j)$ meaning that the game is strongly coupled and that as $\alpha \to 1$, $R_i(a_i,a_j)\to \cR_i(a_i)+ \cR_j(a_j)$ meaning that the game is decoupled (the agents have no effect on other agents' rewards). 

In what follows, the payoff matrix of  $\mathfrak{R}_i(a_i,a_j)$ is denoted by $A$:
This represents the coupled part of the reward in \eqref{reward_interaction_ablation}, i.e. $\mathfrak{R}_i(a_i,a_j)$.  We construct a second matrix corresponding to the independent part of the reward in \eqref{reward_interaction_ablation} and denote this matrix by $B$. Notice in this payoff matrix, the actions of the other agents have no effect on the agent's own reward (whenever an agent plays an action its reward is identical regardless of the other agents action). Thus, to construct the matrix game corresponding to \eqref{reward_interaction_ablation}, we simply compute the weighted sum entry-wise. Denote this by $C$. Now we vary the value of $\alpha$ within the interval $[0,1]$ and plot the number of CL calls used during training (this is the number of times the {\fontfamily{cmss}\selectfont Global} agent performs a switch) vs the value of $\alpha$. 

\begin{table}[h!]
\begin{subtable}{0.3\textwidth}
\centering
\[A=\begin{tabular}[t]{|m{1cm}| |m{1cm}|m{1cm}|}
\hline
&U &  D\\
\hline
U&$5$,$5$ &  $0$,$0$\\
\hline
D& $0$,$0$ & \textcolor{black}{$10$,$10$}\\
\hline
\end{tabular}\]
\end{subtable}
\begin{subtable}{0.32\textwidth}
\centering
\[B=\begin{tabular}[t]{|m{1cm}| |m{1cm}|m{1cm}|}
\hline
&U &  D\\
\hline
U&$10$,$10$ &  $10$,$10$\\
\hline
D&$10$,$10$ & $10$,$10$\\
\hline
\end{tabular}\]
\end{subtable}
\begin{subtable}{0.25\textwidth}
\centering
\[C=\begin{tabular}[t]{|c| |c|c|}
\hline
&U &  D\\
\hline
U&$5(1+\alpha),5(1+\alpha)$ &  $10\alpha,10\alpha$\\
\hline
D&$10\alpha,10\alpha$ & \textcolor{black}{$10,10$}\\
\hline
\end{tabular}\]
\end{subtable}
\end{table}

\newpage\section{Ablation studies}\label{sec:sc_ablation}
\subsection{A.1 Switching Cost Parameter}
\begin{figure}[h]
    \centering
    \includegraphics[width=0.6\textwidth]{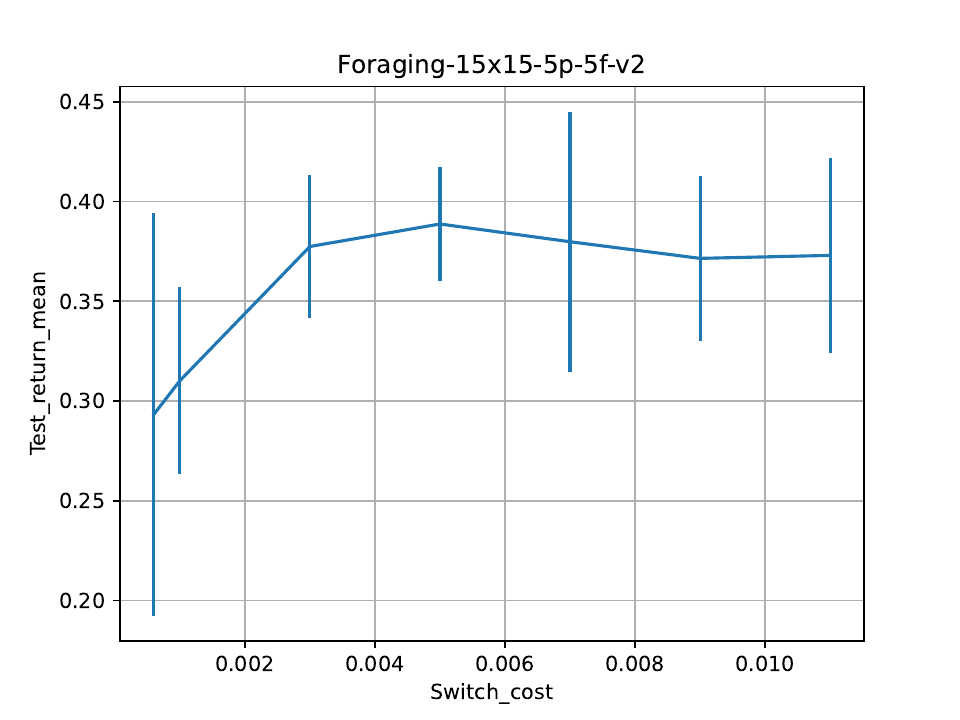}
    \label{fig:ablate}
\end{figure}
We ran a simple study to ascertain the sensitivity of MANSA to the \emph{switching cost} parameter. We picked a random map in LBF and ran MANSA with a range of values for the switching cost. As shown in the plot, while within a given order of magnitude (here $10^{-2}$), MANSA's performance is largely robust to the switching cost. However, there is a deterioration in performance if the switching cost is set too low, and thereby does not penalize usage of CL enough.

\newpage

\newpage\section{MANSA CL Call Analysis}\label{sec:CL_calls_analysis}
One of our key claims is that MANSA reduces the number of CL calls made during training. In this section, we present the CL call percentages made by MANSA in both Level-based foraging (LBF) and StarCraft Multi-Agent Challenge (SMAC). In the case of LBF, MANSA successfully solves the tasks (recall that  MANSA outperforms all baselines in all SMAC maps except \emph{3s5z\_vs\_3s6z} and MANSA outperforms both IQL and QMIX by a notable margin in most of the maps (4 of 8) in LBF, see Section \ref{sec:performances_appendix} for detailed performance plots).     
\begin{center}
\begin{table}[h!]
\centering
 
    \begin{tabular}{c|c} 
        \toprule
        LBF Map& Percentage of CL calls\\
        \midrule
        Foraging-8x8-3p-3f-v2 & 4.76\% \\
        Foraging-10x10-3p-5f-v2 & 5.23\% \\
        Foraging-10x10-5p-3f-v2 & 1.85\% \\
        Foraging-15x15-5p-5f-v2 & 0.76\% \\
        Foraging-5x5-2p-1f-coop-v2 & 3.46\% \\
        Foraging-8x8-2p-2f-coop-v2 & 16.90\% \\
        Foraging-10x10-5p-1f-coop-v2 & 0.71\% \\
        Foraging-10x10-8p-1f-coop-v2 & 0.16\% \\
      \bottomrule
    \end{tabular}
              \caption{Percentage of calls to CL in MANSA in LBF. }
    \label{table:calls_to_cl_LBF}
    \end{table}
\end{center}

\begin{table}[h!]
\centering
    \begin{tabular}{c|c} 
        \toprule
       SMAC Map& Percentage of CL calls\\
        \midrule       
       1c3s5z & 81.67\% \\
       2m\_vs\_1z & 69.01\% \\
       2s3z & 79.70\% \\
       3m & 59.22\% \\
       3s5z & 82.19\% \\
       8m & 62.96\% \\
       corridor & 80.19\% \\
       MMM2 & 80.78\% \\
       so\_many\_baneling & 74.83\% \\
        \bottomrule
    \end{tabular}
        \caption{Percentage of calls to CL in MANSA in SMAC. }
        \label{table:calls_to_cl_SMAC}
    \end{table}
\subsection{MANSA-B CL calls under Budgetary Constraints}

\begin{table}[h!]
    \centering
    \begin{tabular}{p{0.75cm}|p{1.9cm}|p{0.75cm}|p{0.75cm}|p{0.75cm}|p{0.75cm}}
        \toprule
        &\makecell{\textcolor{blue}{Original}/\\\textcolor{orange}{QMIX}/\textcolor{green}{IQL}} & 10\% & 20\% & 50\% & 75\% \\
        \midrule
        3m & \textcolor{blue}{$98.00\pm1.00$} \textcolor{orange}{$92.00\pm1.63$}  \textcolor{green}{$87.00\pm0.82$} & \makecell{$98.00$ \\$\pm 1.00$} & \makecell{$97.67$\\$\pm1.15$} & \makecell{$99.33$\\$\pm0.58$} & \makecell{$99.00$\\$\pm1.00$}\\
        \hline 
        2s3z & \textcolor{blue}{$97.00\pm2.00$}  \textcolor{orange}{$96.33\pm0.58$}  \textcolor{green}{$77.67\pm5.86$} & \makecell{$92.33$\\$\pm5.13$} & \makecell{$96.00$\\$ \pm 3.46$} & \makecell{$90.00$\\$ \pm 5.29$} & \makecell{$96.33$\\$\pm0.57$} \\
        \hline
        2m vs 1z & \textcolor{blue}{$99.00\pm0.00$}  \textcolor{orange}{$68.33\pm28.75$}  \textcolor{green}{$99.00\pm0.82$} & \makecell{$100.00$\\$\pm0.00$} & \makecell{$100.00$\\$\pm0.00$} & \makecell{$100.00$\\$\pm0.00$} & \makecell{$99.67$\\$\pm1.00$} \\
        \hline 
        so many baneling & \textcolor{blue}{$97.00\pm1.00$}  \textcolor{orange}{$86.00\pm3.61$}  \textcolor{green}{$92.34\pm5.03$} & \makecell{$84.50$\\$\pm11.84$} & \makecell{$98.00$\\$\pm2.51$} & \makecell{$95.50$\\$ \pm 2.64$} & \makecell{$93.50$\\$\pm5.13$} \\
        \bottomrule
    \end{tabular}
        \caption{End-of-training win-rates of MANSA-B under various CL call budget constraints. }
            \label{table:budget_mansa}
    \end{table}
End-of-training win-rates of MANSA-B under various CL call budget constraints. Here, the percentages shown on the top row indicate CL calls proportionate to the number of CL calls that was made by MANSA (blue), e.g., 10\% means we only allow MANSA-B total number of CL calls equal to 10\% of the calls of MANSA. The performance of QMIX (orange) and IQL (green) are also shown for reference. In this table we see further evidence of MANSA's remarkably granular control over using CL. In general, performance improves with  each budget increment of CL calls. 
\newpage
\section{MANSA is a Plug \& Play Enhancement Tool}\label{sec:appendix_plugplay}
\begin{table}[h!]
\centering
    \caption{Percentage of CL calls in MANSA on the maps shown in Fig. \ref{figure:further_lbf_exps}. }
    \begin{tabular}{c|c} 
        \toprule
        LBF Map& Percentage of CL (W-QMIX) calls\\
        \midrule
        Foraging-15x15-5p-3f-v2 & 53.15\% \\
        Foraging-15x15-5p-5f-v2 & 87.91\% \\
        Foraging-grid-2s-10x10-3p-3f-v2 & 5.86\% \\
      \bottomrule
    \end{tabular}
    \label{table:calls_to_cl_2}
\end{table}
In this experiment, we replaced QMIX in MANSA with a stronger CL component, W-QMIX, to test if MANSA is able to successfully delivery performance benefits even if the CL baseline is stronger than the IL baseline. As previously discussed in Section \ref{sec:plugnplay}, IQL is outperformed by a CL algorithm, W-QMIX. and in all maps, MANSA significantly outperforms the baselines. From Table \ref{table:calls_to_cl_2} we also see that MANSA uses CL much more in these maps than the maps indicated in Table \ref{table:calls_to_cl_LBF}. Moreover, as with previous experiments, MANSA seems to have correctly identified states that benefit from CL (and those that do not) and have only used CL to achieve significant performance gains.

\newpage

\clearpage
\section{Detailed Performance Plots}\label{sec:performances_appendix}
\subsection{Level-Based Foraging}

\begin{figure}[h]
    \begin{subfigure}{.35\textwidth}
        \centering
        \includegraphics[width=\textwidth]{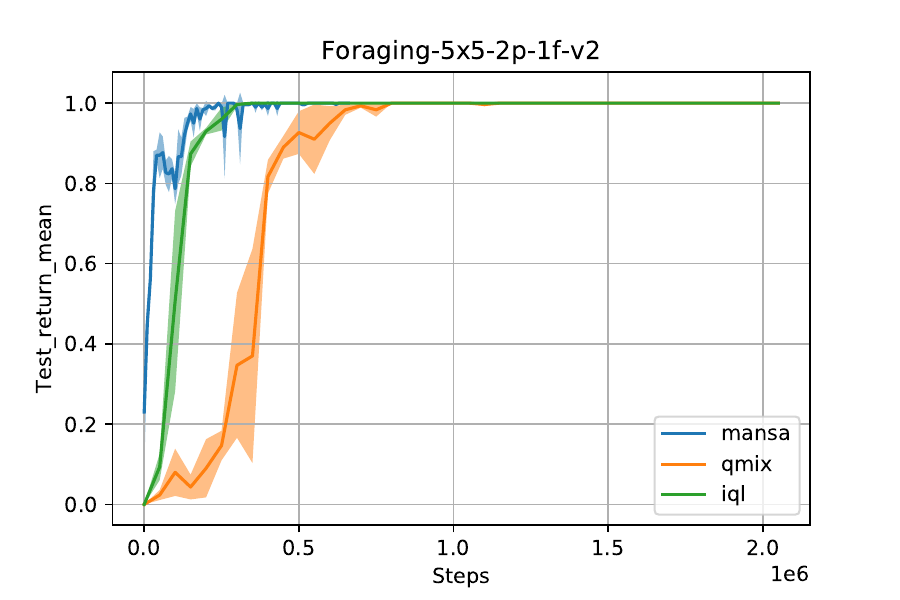}
    \end{subfigure}
    \begin{subfigure}{.35\textwidth}
        \centering
        \includegraphics[width=\linewidth]{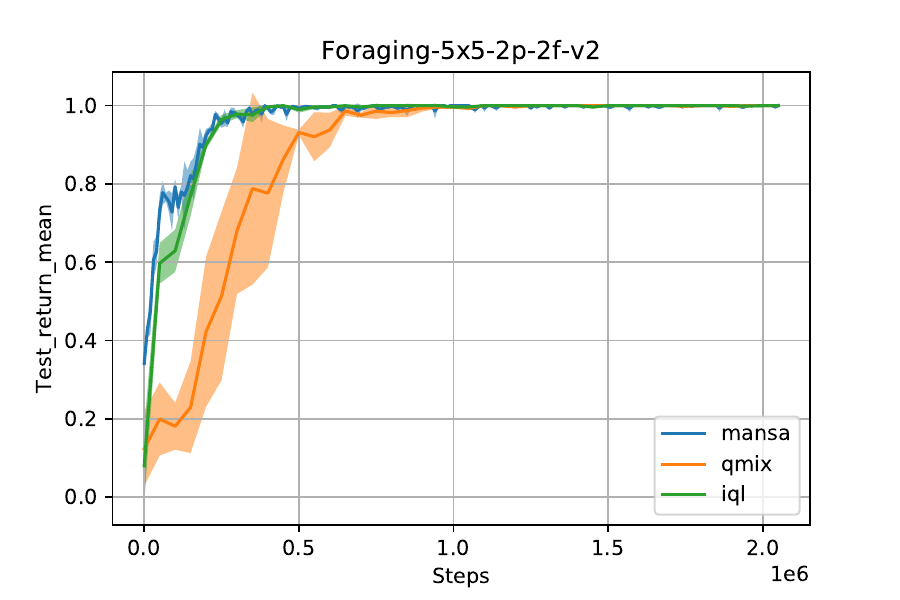}
    \end{subfigure}
    \begin{subfigure}{.35\textwidth}
        \centering
        \includegraphics[width=\linewidth]{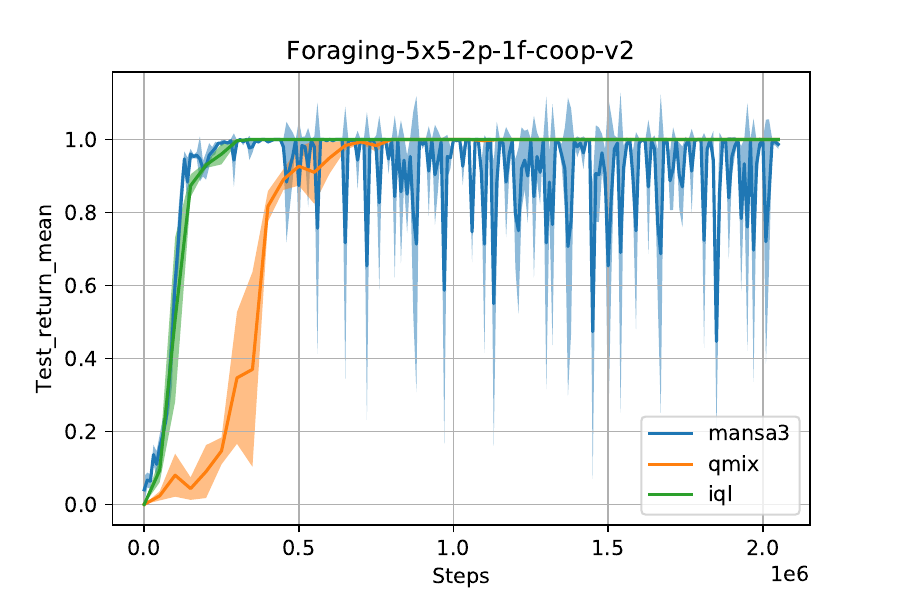}
    \end{subfigure}
    \begin{subfigure}{.35\textwidth}
        \centering
        \includegraphics[width=\linewidth]{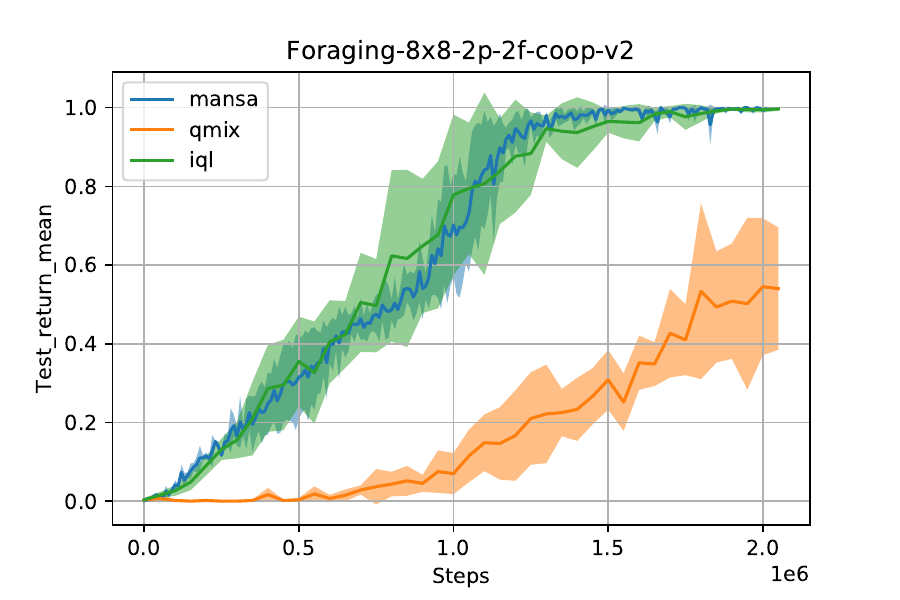}
    \end{subfigure}
    \begin{subfigure}{.35\textwidth}
        \centering
        \includegraphics[width=\linewidth]{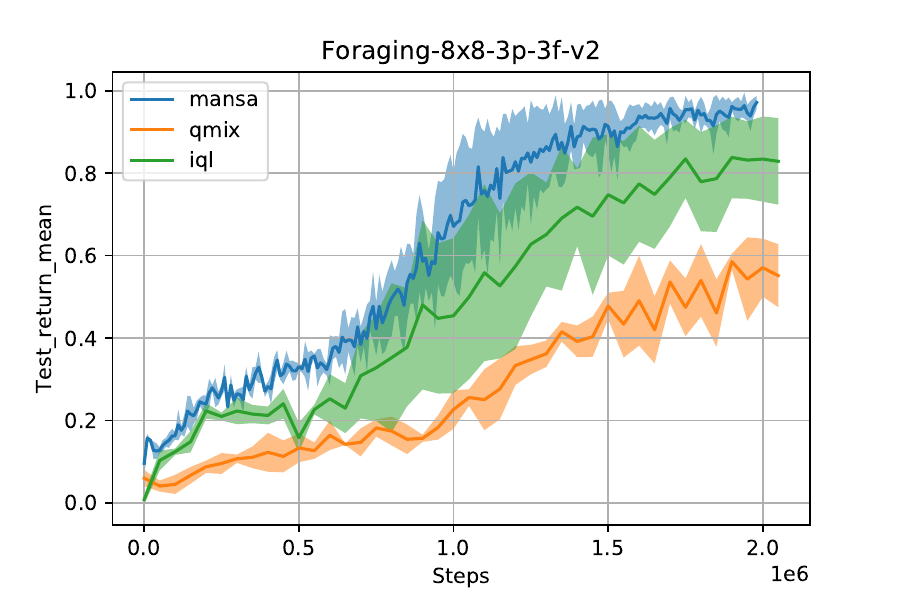}
    \end{subfigure}
    \begin{subfigure}{.35\textwidth}
        \centering
        \includegraphics[width=\linewidth]{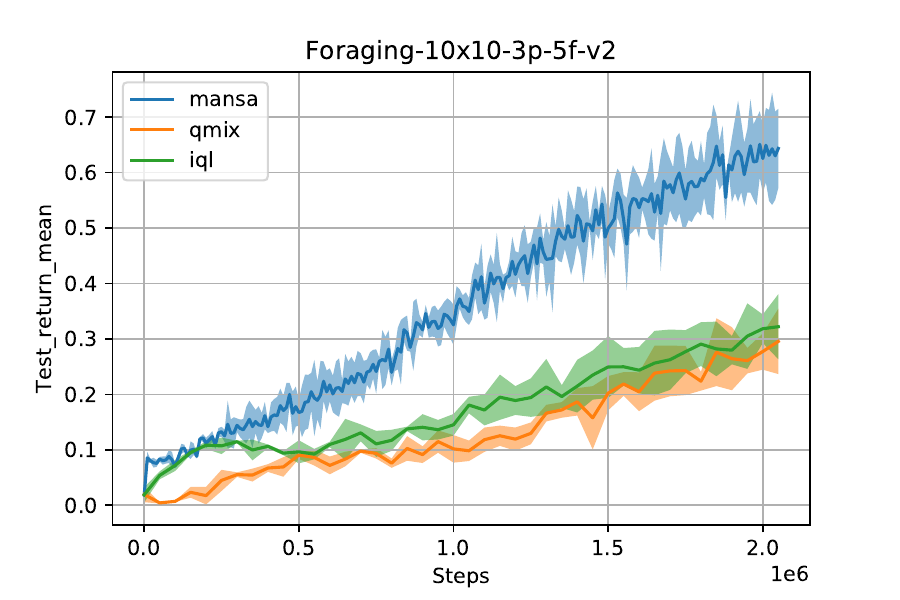}
    \end{subfigure}
    \begin{subfigure}{.35\textwidth}
        \centering
        \includegraphics[width=\textwidth]{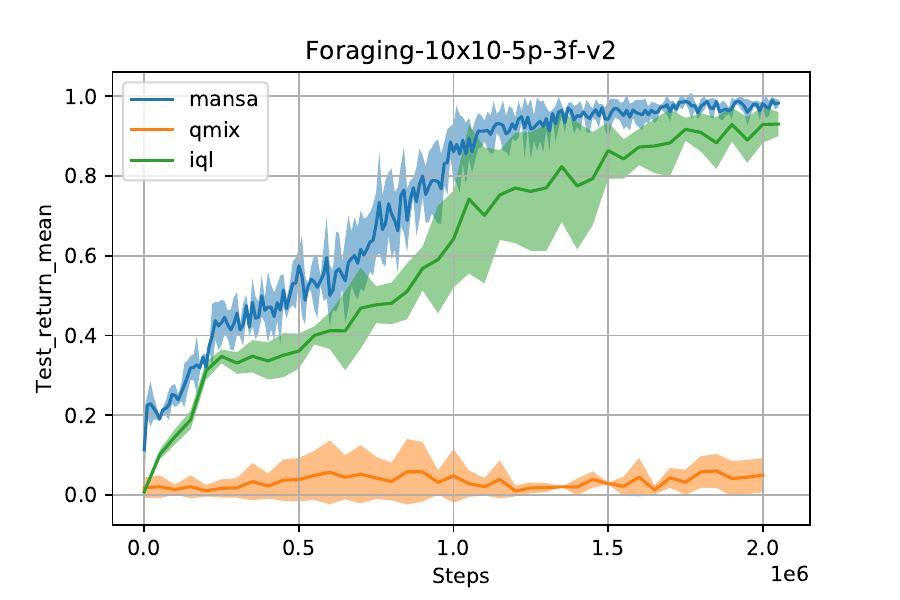}
    \end{subfigure}
    \begin{subfigure}{.35\textwidth}
        \centering
        \includegraphics[width=\linewidth]{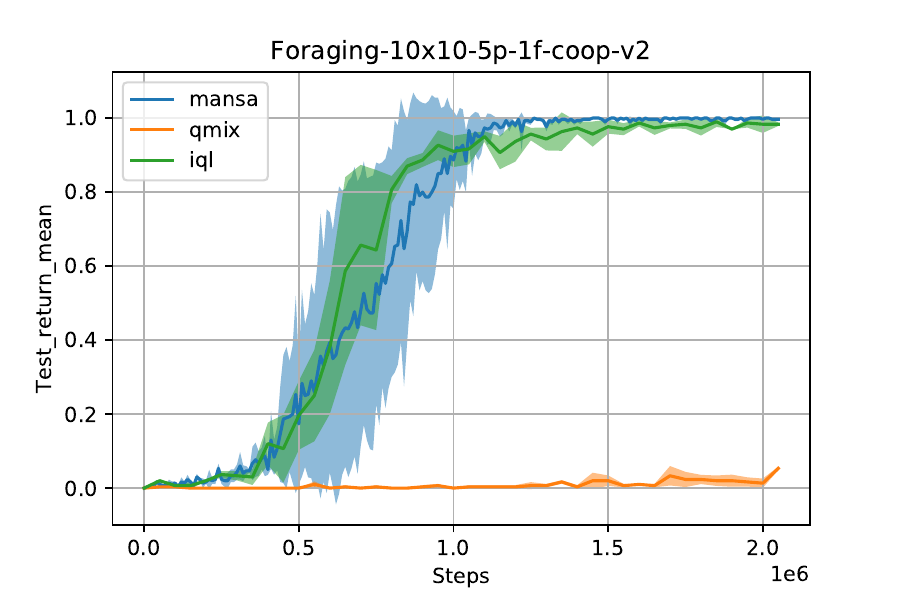}
    \end{subfigure}
    \begin{subfigure}{.35\textwidth}
        \centering
        \includegraphics[width=\linewidth]{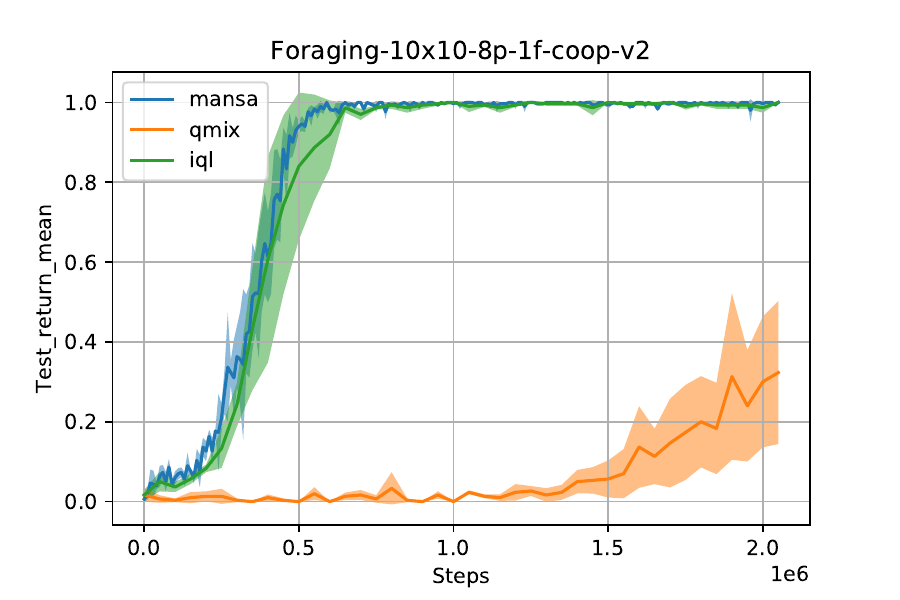}
    \end{subfigure}
    \hfill
    \begin{subfigure}{.35\textwidth}
        \centering
        \includegraphics[width=\linewidth]{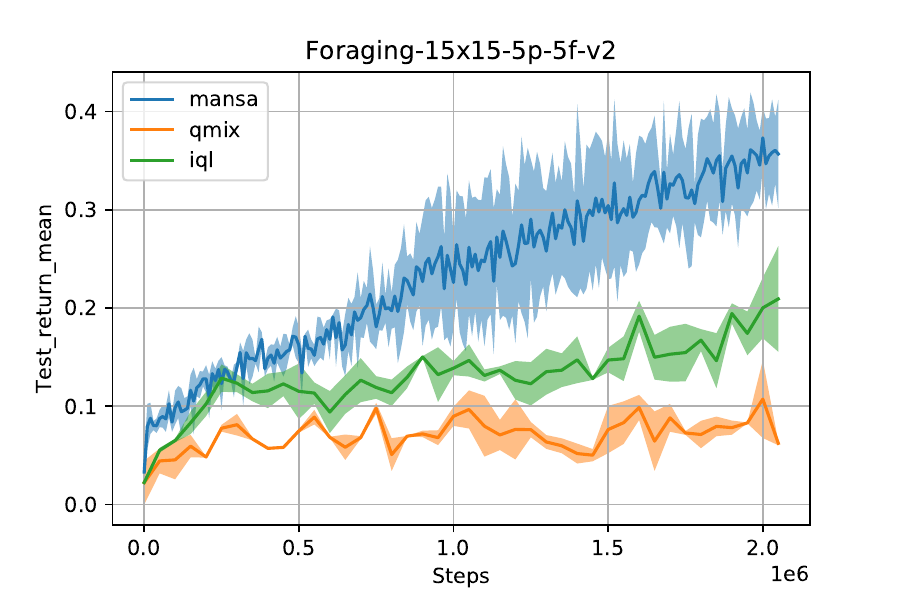}
    \end{subfigure}
    \hfill
\caption{Learning curves on individual LBF maps. }
\end{figure}

\newpage
\subsection{StarCraft Multi-Agent Challenge}
\begin{figure}[h!]
    \centering
    \begin{subfigure}{.35\textwidth}
        \centering
        \includegraphics[width=\textwidth]{Figures/SMAC_experiments/mansa_qmix_iql_/3m.pdf}
    \end{subfigure}\hfill
    \begin{subfigure}{.35\textwidth}
        \centering
        \includegraphics[width=\textwidth]{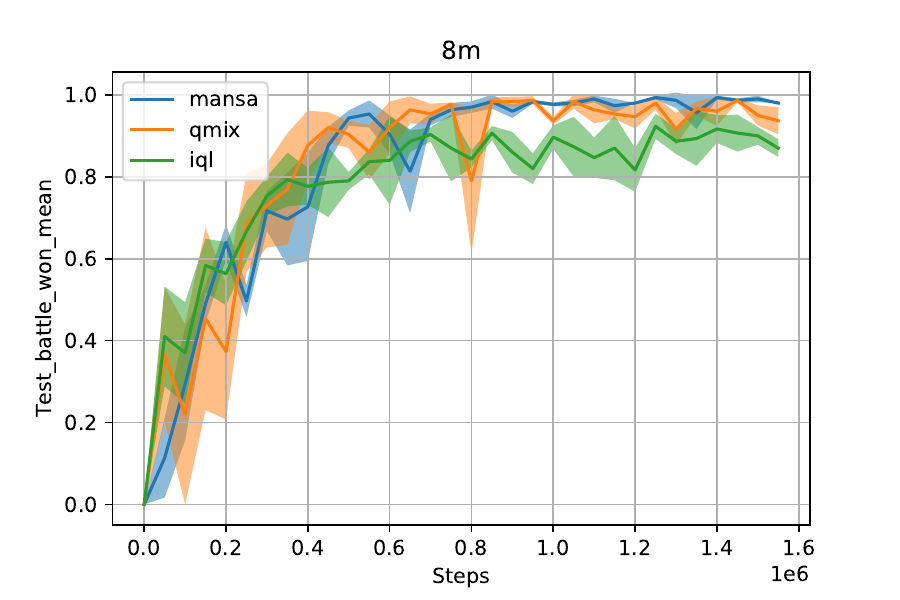}
    \end{subfigure}\hfill
    \begin{subfigure}{.35\textwidth}
        \centering
        \includegraphics[width=\textwidth]{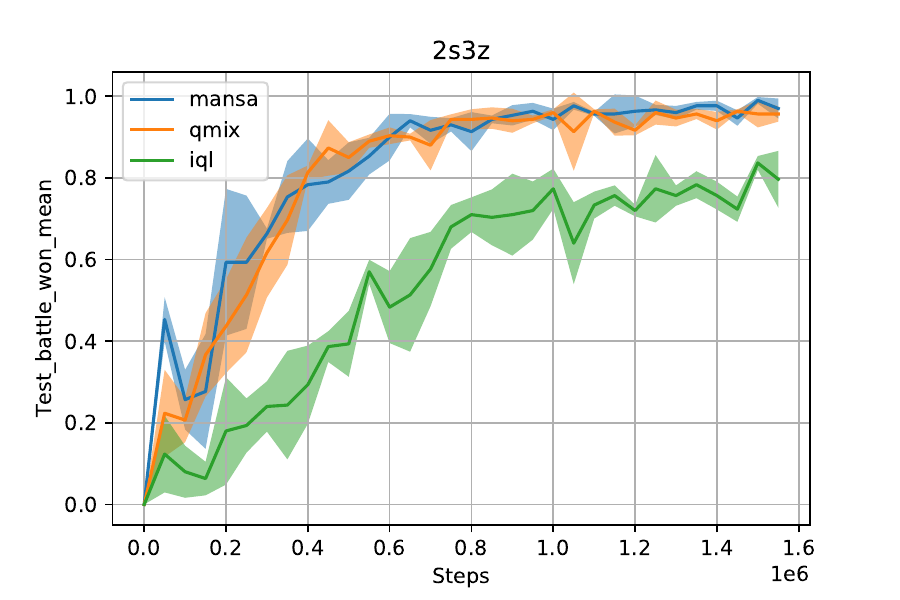}
    \end{subfigure}\hfill
    \begin{subfigure}{.35\textwidth}
        \centering
        \includegraphics[width=\textwidth]{Figures/SMAC_experiments/mansa_qmix_iql_/2m_vs_1z.pdf}
    \end{subfigure}\hfill
    \begin{subfigure}{.35\textwidth}
        \centering
        \includegraphics[width=\textwidth]{Figures/SMAC_experiments/mansa_qmix_iql_/1c3s5z.pdf}
    \end{subfigure}\hfill
    %
    \begin{subfigure}{.35\textwidth}
        \centering
        \includegraphics[width=\textwidth]{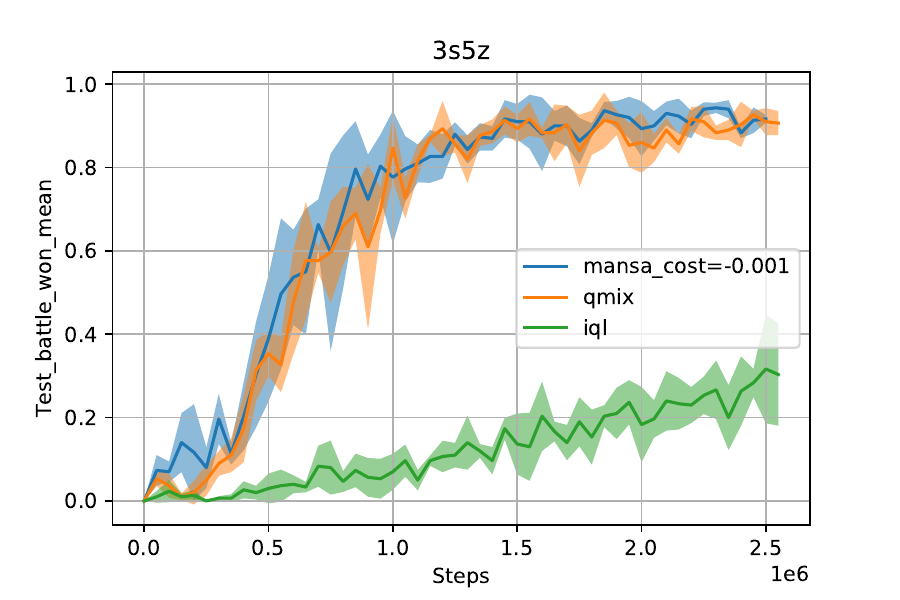}
    \end{subfigure}\hfill
    \begin{subfigure}{.35\textwidth}
        \centering
        \includegraphics[width=\textwidth]{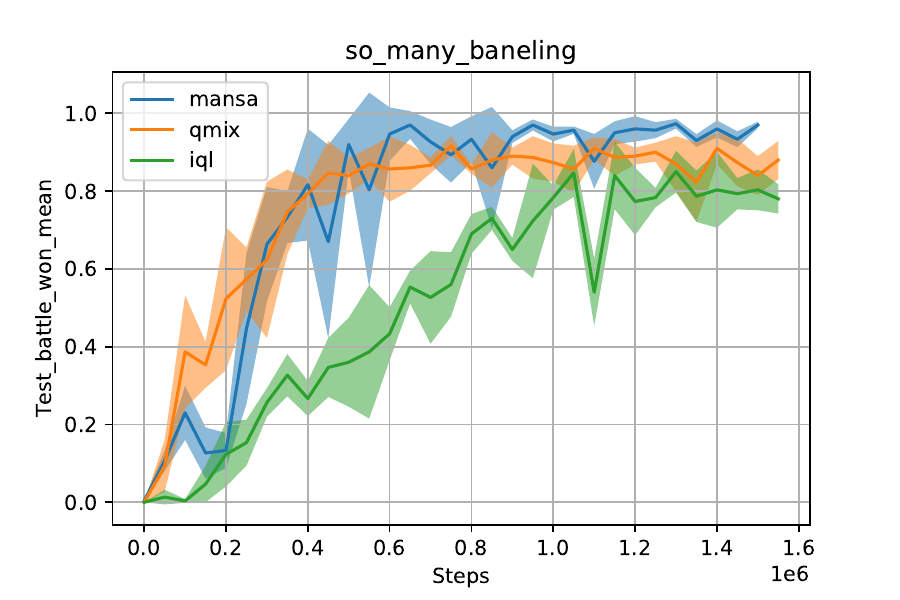}
    \end{subfigure}\hfill
    \begin{subfigure}{.35\textwidth}
        \centering
        \includegraphics[width=\textwidth]{MANSA_review_update/extra_corridor.pdf}
    \end{subfigure}
    \begin{subfigure}{.35\textwidth}
        \centering
        \includegraphics[width=\textwidth]{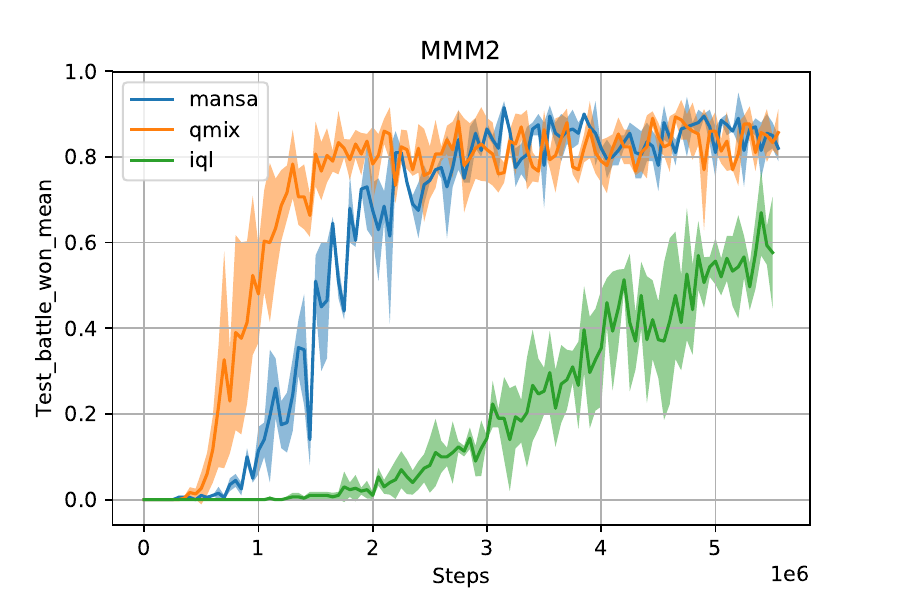}
    \end{subfigure}\hfill
\caption{Learning curves on individual SMAC maps. While QMIX fail to learn effective policies on two maps, and IQL fails on four maps, MANSA does not exhibit any failure cases. We define failure as achieving a win-rate of less than 80\%.}
\end{figure}
\newpage

\section{\textcolor{black}{MANSA with CL Update Restriction}}\label{sec:MANSA_update_restriction}
\textcolor{black}{MANSA includes a feature that imposes the condition that CL updates can only occur when the Global agent makes a CL call (i.e. when $g=1$).  In this section we provide training plots display the results for MANSA with this CL training restriction (MANSA\_CLR) against the baselines. As before, MANSA\_CLR substantially outperforms the baselines on all tested LBF tasks. Similarly, in SMAC, MANSA\_CLR outperforms the baselines on the majority tasks and matches their performance on others. }

\begin{figure}[h!]
    \begin{subfigure}[bh]{\textwidth}
\includegraphics[width=0.45\textwidth]{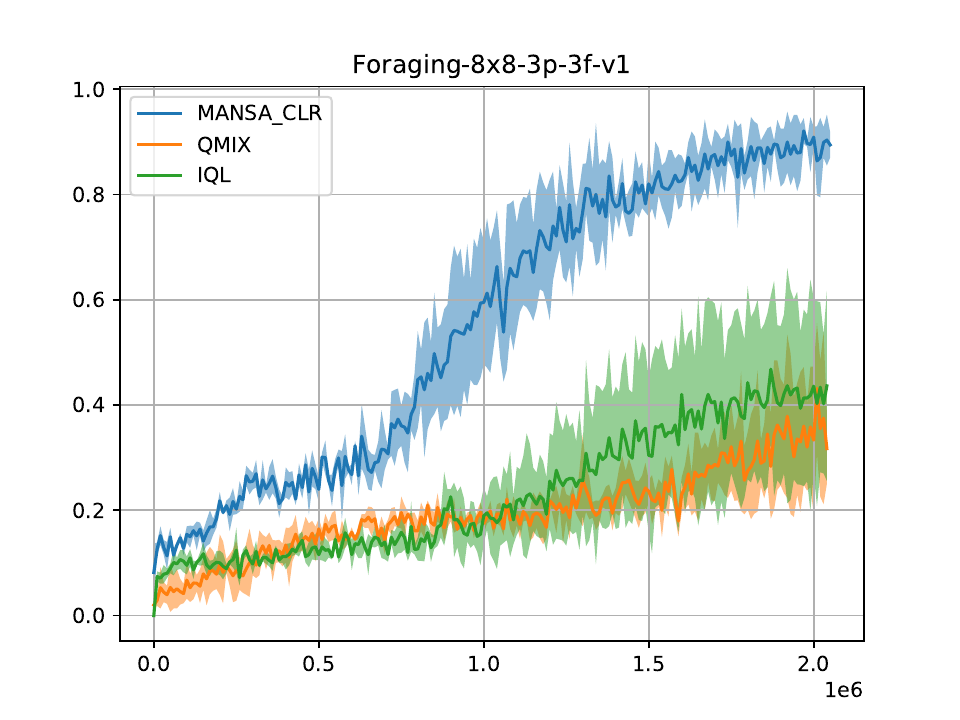}
\includegraphics[width=0.5\textwidth]{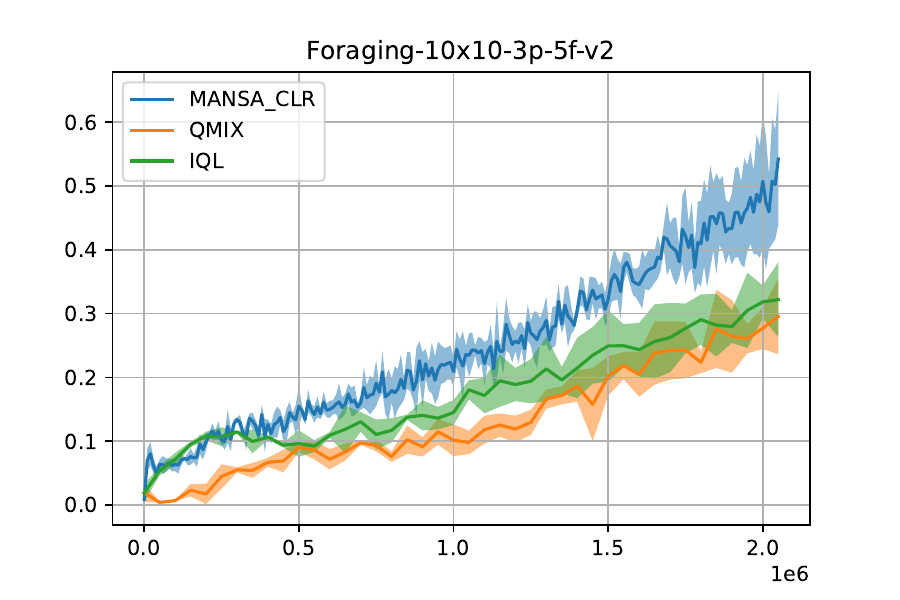}
\end{subfigure}
    \begin{subfigure}[b]{\textwidth}
\includegraphics[width=0.5\textwidth]{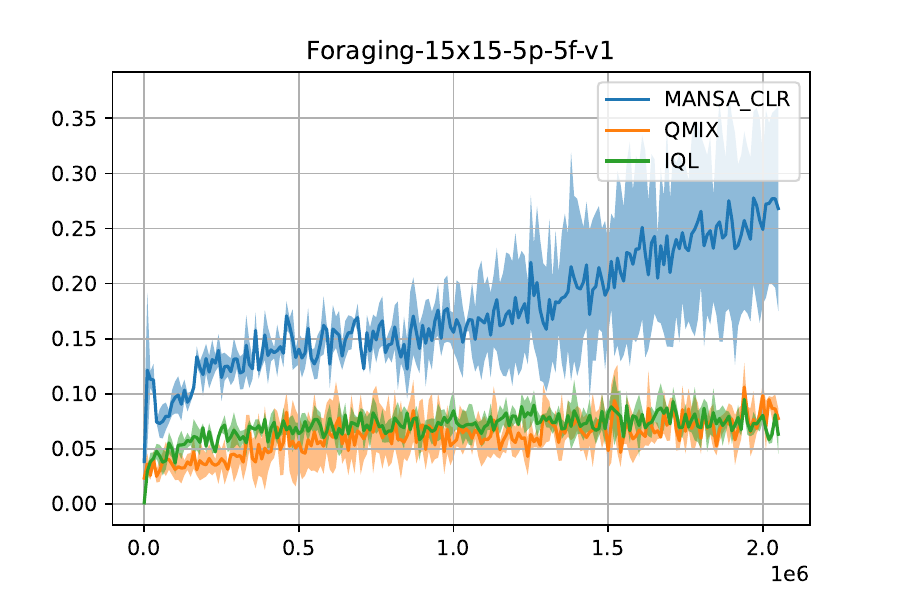}
\end{subfigure}
        \caption{\textcolor{black}{End-of-training returns of MANSA with CL update restriction (MANSA\_CLR) in Level-Based Foraging (LBF).}}
    \label{fig:returns_CLR}
\end{figure}

\begin{figure}[ht!]
    \begin{subfigure}[b]{\textwidth}
\includegraphics[width=0.5\textwidth]{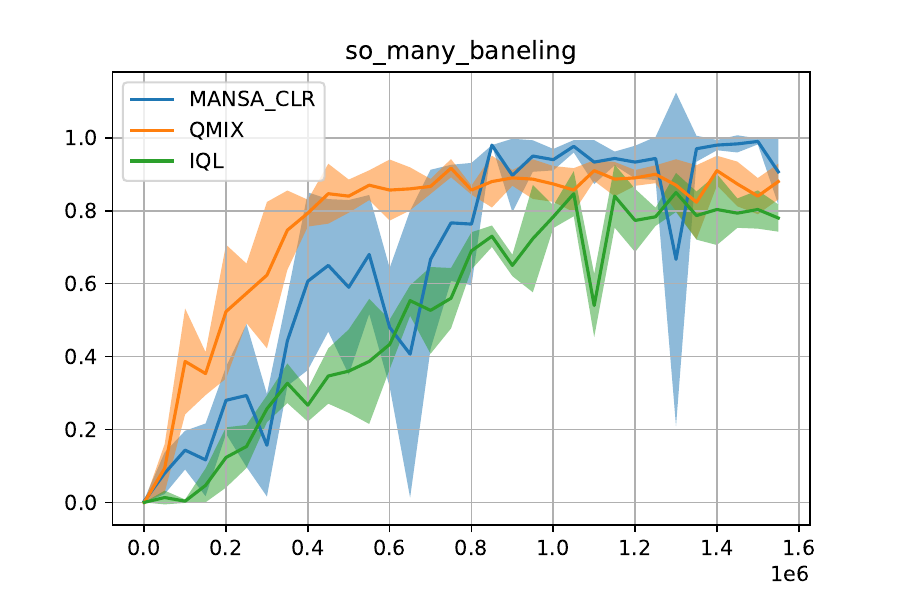}
    \includegraphics[width=0.5\textwidth]{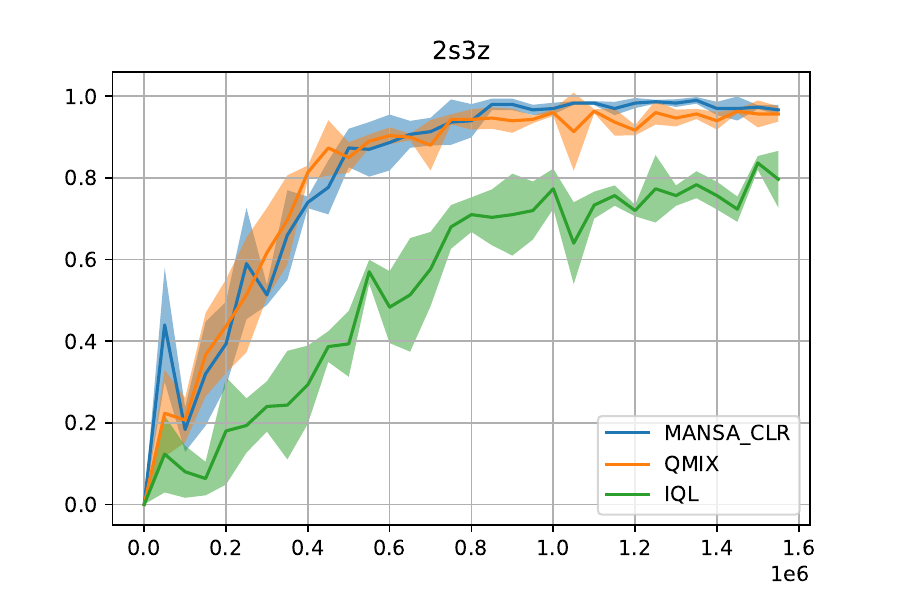}
\end{subfigure}
    \begin{subfigure}[b]{\textwidth}
\includegraphics[width=0.5\textwidth]{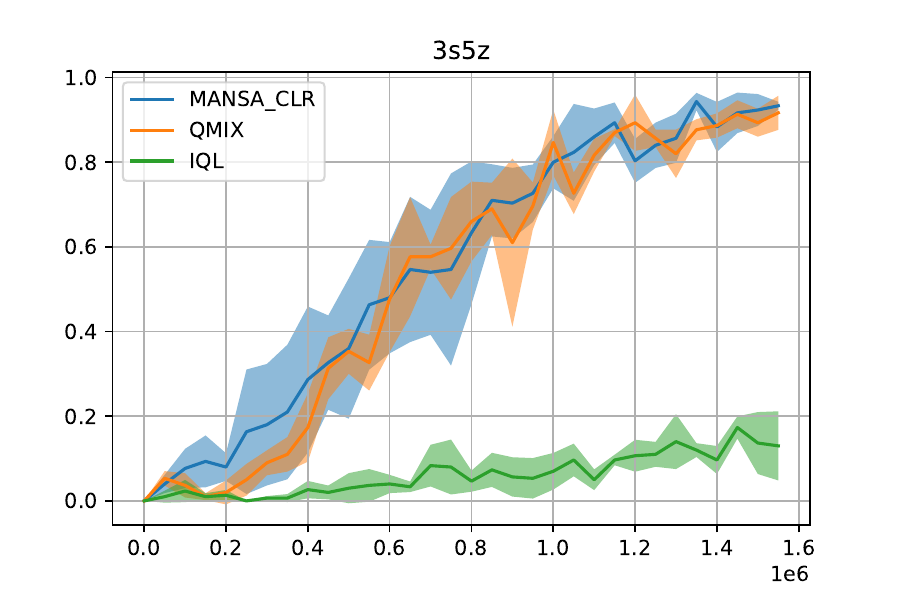}
\end{subfigure}
        \caption{\textcolor{black}{End-of-training win-rates of MANSA with implementation with CL update restriction (MANSA\_CLR) in StarCraft Multi-Agent Challenge (SMAC).}}
    \label{fig:}
\end{figure}
\clearpage
\subsection{\textcolor{black}{\textbf{MANSA-B}udget with CL Update Restriction}}
\textcolor{black}{In Section \ref{sec:MANSA_update_restriction}. we provided results for MANSA\_CLR which imposes the restriction that CL updates can only occur when the Global agent makes a CL call (i.e. when $g=1$). In this section, Table \ref{table:budget_mansa_CLR} displays the results for MANSA-B\_CLR which imposes the CL update restriction on the MANSA-B framework (i.e. MANSA which has a budget constraint on the number of CL calls). As before, MANSA\_CLR outperforms IQL when given a budget of just 20\% CL calls and outperforms QMIX on 2m\_vs\_1z with just a 10\% CL budget. In 2s3z MANSA outperforms QMIX when it has a budget of 75\% for its CL calls; i.e. it outperforms QMIX even though it is forced to make 25\% fewer CL calls than QMIX.} 
\begin{table}[bh!]
    \centering
    \begin{center}
    \begin{tabular}{p{2.0cm}|p{3.5cm}|p{0.75cm}|p{0.75cm}|p{0.75cm}|p{0.75cm}}
       \toprule
        &\makecell{\textcolor{blue}{\textbf{Original}}\newline/\textcolor{orange}{\textbf{QMIX}}\newline/\textcolor{green}{\textbf{IQL}}} & \textbf{10\%} & \textbf{20\%} & \textbf{50\%} & \textbf{75\%} 
        \\\midrule
        \makecell{\textbf{2m\_vs\_1z}} & \textcolor{blue}{\;\;\;\;\;$98.00\pm1.00$}\newline \textcolor{orange}{\;\;\;\;\;$92.00\pm1.63$}\newline  \textcolor{green}{\;\;\;\;\;$87.00\pm0.82$} & \makecell{\hspace{-1.5 mm}$100.00$ \\$\pm 0.00$} & \makecell{$99.67$\\$\pm0.57$} & \makecell{$96.67$\\$\pm3.05$} & \makecell{$99.00$\\$\pm0.00$}\\
        \hline
        \makecell{\textbf{2s3z}} & \textcolor{blue}{\;\;\;\;\;$96.67\pm1.24$}\newline \textcolor{orange}{\;\;\;\;\;$95.67\pm1.8$}\newline  \textcolor{green}{\;\;\;\;\;$79.67\pm6.69$} & \makecell{\hspace{-1.5 mm}$82.00$ \\$\pm 1.41$} & \makecell{$82.33$\\$\pm5.18$} & \makecell{$81.67$\\$\pm1.69$} & \makecell{$96.33$\\$\pm0.47$}\\
        \bottomrule
    \end{tabular}
    \end{center}
        \caption{\textcolor{black}{End-of-training win-rates of MANSA-B with CL update restriction and various CL call budget constraints against baselines.}}
    \label{table:budget_mansa_CLR}
\end{table}

\newpage
\section{Hyperparameter Settings}
In the table below we report all hyperparameters used in our experiments. Hyperparameter values in square brackets indicate ranges of values that were used for performance tuning.

\begin{center}
    \begin{tabular}{c|c} 
        \toprule
        Clip Gradient Norm & 1\\
        $\gamma_{E}$ & 0.99\\
        $\lambda$ & 0.95\\
        Learning rate & $1$x$10^{-4}$ \\
        Number of minibatches & 4\\
        Number of optimisation epochs & 4\\
        Number of parallel actors & 16\\
        Optimisation algorithm & Adam\\
        Rollout length & 128\\
        Sticky action probability & 0.25\\
        Use Generalized Advantage Estimation & True\\
        \midrule
        Coefficient of extrinsic reward & [1, 5]\\
        Coefficient of intrinsic reward & [1, 2, 5, 10, 20, 50]\\
        {\fontfamily{cmss}\selectfont Global} discount factor & 0.99\\
        Probability of terminating option & [0.5, 0.75, 0.8, 0.9, 0.95]\\
        $L$ function output size & [2, 4, 8, 16, 32, 64, 128, 256]\\
        \bottomrule
    \end{tabular}
\end{center}

\clearpage
\clearpage
\section{Notation \& Assumptions}\label{sec:notation_appendix}

We assume that $\mathcal{S}$ is defined on a probability space $(\Omega,\mathcal{F},\mathbb{P})$ and any $s\in\mathcal{S}$ is measurable with respect
to the Borel $\sigma$-algebra associated with $\mathbb{R}^p$. We denote the $\sigma$-algebra of events generated by $\{s_t\}_{t\geq 0}$
by $\mathcal{F}_t\subset \mathcal{F}$. In what follows, we denote by $\left( \cY,\|\|\right)$ any finite normed vector space and by $\mathcal{H}$ the set of all measurable functions.  Where it will not cause confusion (and with a minor abuse of notation) for a given function $h$ we use the shorthand $h^{(\pi^{i},\pi^{-i})}(s)= h(s,\pi^i,\pi^{-i})\equiv\mathbb{E}_{\pi^i,\pi^{-i}}[h(s,a^i,a^{-i})]$.

The results of the paper are built under the following assumptions which are standard within RL and stochastic approximation methods:

\textbf{Assumption 1}
The stochastic process governing the system dynamics is ergodic, that is  the process is stationary and every invariant random variable of $\{s_t\}_{t\geq 0}$ is equal to a constant with probability $1$.

\textbf{Assumption 2}
The agents' reward function $R$ is in $L_2$.

\textbf{Assumption 3}
For any positive scalar $c$, there exists a scalar $\mu_c$ such that for all $s\in\mathcal{S}$ and for any $t\in\mathbb{N}$ we have: $
    \mathbb{E}\left[1+\|s_t\|^c|s_0=s\right]\leq \mu_c(1+\|s\|^c)$.

\textbf{Assumption 4}
There exists scalars $C_1$ and $c_1$ such that  $|R(s,\cdot)|\leq C_2(1+\|s\|^{c_2})$ for some scalars $c_2$ and $C_2$ we have that: $
    \sum_{t=0}^\infty\left|\mathbb{E}\left[R(s_t,\cdot)|s_0=s\right]-\mathbb{E}[R(s_0,\cdot)]\right|\leq C_1C_2(1+\|s_t\|^{c_1c_2})$.

\textbf{Assumption 5}
There exists scalars $e$ and $E$ such that for any $s\in\mathcal{S}$ we have that: $
    |R(s,\cdot)|\leq E(1+\|s\|^e)$ .

\textbf{Assumption 6}
For any {\fontfamily{cmss}\selectfont Global} policy $\mathfrak{g}$, the total number of interventions is $K<\infty$.

\newpage\section{Proof of Technical Results}\label{sec:proofs_appendix}

We begin the analysis with some preliminary results and definitions required for proving our main results.

\begin{definition}{A.1}
Given a norm $\|\cdot\|$, an operator $T: \cY\to \cY$ is a contraction if there exists some constant $c\in[0,1[$ for which for any $J_1,J_2\in  \cY$ the following bound holds: $    \|TJ_1-TJ_2\|\leq c\|J_1-J_2\|$.
\end{definition}

\begin{definition}{A.2}
An operator $T: \cY\to  \cY$ is non-expansive if $\forall J_1,J_2\in  \cY$ the following bound holds: $    \|TJ_1-TJ_2\|\leq \|J_1-J_2\|$.
\end{definition}

\begin{lemma} \cite{mguni2019cutting} \label{max_lemma}
For any
$f: \cY\to\mathbb{R}: \cY\to\mathbb{R}$, we have that the following inequality holds:
\begin{align}
\left\|\underset{a\in \cY}{\max}\:f(a)-\underset{a\in \cY}{\max}\: g(a)\right\| \leq \underset{a\in \cY}{\max}\: \left\|f(a)-g(a)\right\|.    \label{lemma_1_basic_max_ineq}
\end{align}
\end{lemma}

\begin{lemma} {A.4}\citep{tsitsiklis1999optimal}\label{non_expansive_P}
The probability transition kernel $P$ is non-expansive so that if $\forall J_1,J_2\in  \cY$ the following holds: $    \|PJ_1-PJ_2\|\leq \|J_1-J_2\|$.
\end{lemma} 

\section*{Proof of Theorem \ref{theorem:existence}}
\begin{proof}

The proof of the Theorem proceeds by first proving that for any two fixed set of joint policies $\boldsymbol{\pi}^d,\boldsymbol{\pi}^c \in \boldsymbol{\Pi}$, the  {\fontfamily{cmss}\selectfont Global} agent's learning process, which involves switching controls converges. Recall, that the {\fontfamily{cmss}\selectfont Global} agent presides over an activation that deactivates $\boldsymbol{\pi}^d$ and activates $\boldsymbol{\pi}^c$.

Prove that the solution to Markov Team games (that is games in which both players maximise \textit{identical objectives}) in which one of the players uses switching control is the limit point of a sequence of Bellman operators (acting on some test function)

Therefore, the scheme of the proof is summarised with the following steps:
\begin{itemize}
    \item[\textbf{A)}] Prove that for any fixed {\fontfamily{cmss}\selectfont Central} and {\fontfamily{cmss}\selectfont Decentral} policies $\boldsymbol{\pi}^c$ and $\boldsymbol{\pi}^d$, {\fontfamily{cmss}\selectfont Global}'s switching control policy converges to a solution of {\fontfamily{cmss}\selectfont Global}'s problem.
    \item[\textbf{B)}] Prove that the MG $\mathcal{G}$ has a dual representation as a \textit{Markov Team Game} whose solution is obtained by computing the solution of a team Markov game. 
    \item[\textbf{C)}]  
Prove that all agents solve the same problem.
\end{itemize}

We begin by recalling the definition of the intervention operator $\mathcal{M}^{\mathfrak{g},\boldsymbol{\pi}^c}$ for any $s\in\cS$ and for a given $\boldsymbol{\pi}^c$: 
\begin{align}
\mathcal{M}^{\mathfrak{g},\boldsymbol{\pi}^c}Q_G(s,\boldsymbol{a}|\cdot):=Q_G(s,\boldsymbol{\pi}^c(s)|\cdot)-c
\end{align}

Secondly, recall that the Bellman operator for the game $\cG$ is given by:

\begin{align}
T_gv_G(s_{\tau_k}):=\max\left\{\mathcal{M}^{\mathfrak{g},\boldsymbol{\pi}^c}Q_G(s_{\tau_k},\boldsymbol{a}),\underset{\boldsymbol{a}\in\boldsymbol{\mathcal{A}}}{\max}\;\left[ R_G(s_{\tau_k},\boldsymbol{a},g)+\gamma\sum_{s'\in\mathcal{S}}P(s';\boldsymbol{a},s_{\tau_k})v_G(s')\right]\right\}\label{bellman_proof_start}
\end{align}

To prove (i) it suffices to prove that $T$ is a contraction operator. Thereafter, we use both results to prove the existence of a fixed point for $\cG$ as a limit point of a sequence generated by successively applying the Bellman operator to a test value function.   
Therefore our next result shows that the following bounds holds:
\begin{lemma}\label{lemma:bellman_contraction}
The Bellman operator $T$ is a contraction so that the following bound holds: $
\left\|T\psi-T\psi'\right\|\leq \gamma\left\|\psi-\psi'\right\|$.
\end{lemma}


In the following proofs we use the following notation: $
\mathcal{P}^{\boldsymbol{a}}_{ss'}=:\sum_{s'\in\mathcal{S}}P(s';\boldsymbol{a},s)$ and $\mathcal{P}^{\boldsymbol{\pi}}_{ss'}=:\sum_{\boldsymbol{a}\in\boldsymbol{\mathcal{A}}}\boldsymbol{\pi}(\boldsymbol{a}|s)\mathcal{P}^{\boldsymbol{a}}_{ss'}$.

To prove that $T$ is a contraction, we consider the three cases produced by \eqref{bellman_proof_start}, that is to say we prove the following statements:

i) $\qquad\qquad
\left| \underset{\boldsymbol{a}\in\boldsymbol{\mathcal{A}}}{\max}\;\left(R_G(s_t,\boldsymbol{a},g)+\gamma\mathcal{P}^{\boldsymbol{a}}_{s's_t}v_G(s')\right)-\underset{\boldsymbol{a}\in\boldsymbol{\mathcal{A}}}{\max}\;\left( R_G(s_t,\boldsymbol{a},g)+\gamma\mathcal{P}^{\boldsymbol{a}}_{s's_t}v_G'(s')\right)\right|\leq \gamma\left\|v_G-v_G'\right\|$

ii) $\qquad\qquad
\left\|\mathcal{M}^{\mathfrak{g},\boldsymbol{\pi}^c}Q_G-\mathcal{M}^{\mathfrak{g},\boldsymbol{\pi}^c}Q_G'\right\|\leq    \gamma\left\|v_G-v_G'\right\|,\qquad \qquad$.

iii) $\qquad\qquad
    \left\|\mathcal{M}^{\mathfrak{g},\boldsymbol{\pi}^c}Q_G-\underset{\boldsymbol{a}\in\boldsymbol{\mathcal{A}}}{\max}\;\left[ R_G(s_t,\boldsymbol{a},g)+\gamma\mathcal{P}^{\boldsymbol{a}}v_G'\right]\right\|\leq \gamma\left\|v_G-v_G'\right\|.
$

We begin by proving i).

Indeed, for any $\boldsymbol{a}\in\boldsymbol{\mathcal{A}}$ and $\forall s_t\in\mathcal{S}, \forall s'\in\mathcal{S}$ we have that 
\begin{align*}
&\left| \underset{\boldsymbol{a}\in\boldsymbol{\mathcal{A}}}{\max}\;\left(R_G(s_t,\boldsymbol{a},g)+\gamma\mathcal{P}^\pi_{s's_t}v_G(s')\right)-\underset{\boldsymbol{a}\in\boldsymbol{\mathcal{A}}}{\max}\;\left( R_G(s_t,\boldsymbol{a},g)+\gamma\mathcal{P}^{\boldsymbol{a}}_{s's_t}v_G'(s')\right)\right|
\\&\leq \underset{\boldsymbol{a}\in\boldsymbol{\mathcal{A}}}{\max}\;\left|\gamma\mathcal{P}^{\boldsymbol{a}}_{s's_t}v_G(s')-\gamma\mathcal{P}^{\boldsymbol{a}}_{s's_t}v_G'(s')\right|
\\&\leq \gamma\left\|Pv_G-Pv_G'\right\|
\\&\leq \gamma\left\|v_G-v_G'\right\|,
\end{align*}
using the non-expaniveness of the operator $P$ and Lemma \ref{max_lemma}.

We now prove ii). Using the definition of $\mathcal{M}$ we have that for any $s_\tau\in\mathcal{S}$
\begin{align*}
&\left|(\mathcal{M}^{\mathfrak{g},\boldsymbol{\pi}^c}Q_G-\mathcal{M}^{\mathfrak{g},\boldsymbol{\pi}^c}Q_G')(s_{\tau},\boldsymbol{a}_{\tau})\right|
\\&=\Bigg|R_G(s_\tau,\boldsymbol{\pi}^c,g)- c+\gamma\mathcal{P}^{\boldsymbol{\pi}}_{s's_\tau}\mathcal{P}^{\boldsymbol{\pi}^c}v_G(s_{\tau})
-\left(R_G(s_\tau,\boldsymbol{\pi}^c,g)- c+\gamma\mathcal{P}^{\boldsymbol{\pi}}_{s's_\tau}\mathcal{P}^{\boldsymbol{\pi}^c}v_G'(s_{\tau})\right)\Bigg|
\\&\leq \underset{\boldsymbol{a}_\tau,g\in \boldsymbol{\mathcal{A}}\times \{0,1\}}{\max}    \Bigg|R_G(s_\tau,\boldsymbol{a}_\tau,g)- c+\gamma\mathcal{P}^{\boldsymbol{\pi}}_{s's_\tau}\mathcal{P}^{\boldsymbol{a}}v_G(s_{\tau})
-\left(R_G(s_\tau,\boldsymbol{a}_\tau,g)- c+\gamma\mathcal{P}^{\boldsymbol{\pi}}_{s's_\tau}\mathcal{P}^{\boldsymbol{a}}v_G'(s_{\tau})\right)\Bigg|
\\&= \gamma\underset{\boldsymbol{a}_\tau,g\in \boldsymbol{\mathcal{A}}\times \{0,1\}}{\max}    \Bigg|\mathcal{P}^{\boldsymbol{\pi}}_{s's_\tau}\mathcal{P}^{\boldsymbol{a}}v_G(s_{\tau})
-\mathcal{P}^{\boldsymbol{\pi}}_{s's_\tau}\mathcal{P}^{\boldsymbol{a}}v_G'(s_{\tau})\Bigg|
\\&\leq \gamma\left\|Pv_G-Pv_G'\right\|
\\&\leq \gamma\left\|v_G-v_G'\right\|,
\end{align*}
using the fact that $P$ is non-expansive. The result can then be deduced easily by applying max on both sides.

We now prove iii). We split the proof of the statement into two cases:

\textbf{Case 1:} 
First, assume that for any $s_\tau\in\cS$ and $\forall \boldsymbol{a}\in\boldsymbol{\cA}$ the following inequality holds:
\begin{align}\mathcal{M}^{\mathfrak{g},\boldsymbol{\pi}^c}Q_G(s_{\tau},\boldsymbol{a})-\underset{\boldsymbol{a}\in\boldsymbol{\mathcal{A}}}{\max}\;\left(R_G(s_\tau,\boldsymbol{a}_\tau,g)+\gamma\mathcal{P}^{\boldsymbol{a}}_{s's_\tau}v_G'(s')\right)<0.
\end{align}

We now observe the following:
\begin{align*}
&\mathcal{M}^{\mathfrak{g},\boldsymbol{\pi}^c}Q_G(s_{\tau},\boldsymbol{a})-\underset{\boldsymbol{a}\in\boldsymbol{\mathcal{A}}}{\max}\;\left(R_G(s_\tau,\boldsymbol{a}_\tau,g)+\gamma\mathcal{P}^{\boldsymbol{a}}_{s's_\tau}v_G'(s')\right)
\\&\leq\max\left\{\underset{\boldsymbol{a}\in\boldsymbol{\mathcal{A}}}{\max}\;\left(R_G(s_\tau,\boldsymbol{a}_\tau,g)+\gamma\mathcal{P}^{\boldsymbol{\pi}}_{s's_\tau}\mathcal{P}^{\boldsymbol{a}}v_G(s')\right),\mathcal{M}^{\mathfrak{g},\boldsymbol{\pi}^c}Q_G(s_{\tau},\boldsymbol{a})\right\}
-\underset{\boldsymbol{a}\in\boldsymbol{\mathcal{A}}}{\max}\;\left(R_G(s_\tau,\boldsymbol{a}_\tau,g)+\gamma\mathcal{P}^{\boldsymbol{a}}_{s's_\tau}v_G'(s')\right)
\\&\leq \Bigg|\max\left\{\underset{\boldsymbol{a}\in\boldsymbol{\mathcal{A}}}{\max}\;\left(R_G(s_\tau,\boldsymbol{a}_\tau,g)+\gamma\mathcal{P}^{\boldsymbol{\pi}}_{s's_\tau}\mathcal{P}^{\boldsymbol{a}}v_G(s')\right),\mathcal{M}^{\mathfrak{g},\boldsymbol{\pi}^c}Q_G(s_{\tau},\boldsymbol{a})\right\}
\\&\qquad-\max\left\{\underset{\boldsymbol{a}\in\boldsymbol{\mathcal{A}}}{\max}\;\left(R_G(s_\tau,\boldsymbol{a}_\tau,g)+\gamma\mathcal{P}^{\boldsymbol{a}}_{s's_\tau}v_G'(s')\right),\mathcal{M}^{\mathfrak{g},\boldsymbol{\pi}^c}Q_G(s_{\tau},\boldsymbol{a})\right\}
\\&+\max\left\{\underset{\boldsymbol{a}\in\boldsymbol{\mathcal{A}}}{\max}\;\left(R_G(s_\tau,\boldsymbol{a}_\tau,g)+\gamma\mathcal{P}^{\boldsymbol{a}}_{s's_\tau}v_G'(s')\right),\mathcal{M}^{\mathfrak{g},\boldsymbol{\pi}^c}Q_G(s_{\tau},\boldsymbol{a})\right\}-\underset{\boldsymbol{a}\in\boldsymbol{\mathcal{A}}}{\max}\;\left(R_G(s_\tau,\boldsymbol{a}_\tau,g)+\gamma\mathcal{P}^{\boldsymbol{a}}_{s's_\tau}v_G'(s')\right)\Bigg|
\\&\leq \Bigg|\max\left\{\underset{\boldsymbol{a}\in\boldsymbol{\mathcal{A}}}{\max}\;\left(R_G(s_\tau,\boldsymbol{a}_\tau,g)+\gamma\mathcal{P}^{\boldsymbol{a}}_{s's_\tau}v_G(s')\right),\mathcal{M}^{\mathfrak{g},\boldsymbol{\pi}^c}Q_G(s_{\tau},\boldsymbol{a})\right\}
\\&\qquad-\max\left\{\underset{\boldsymbol{a}\in\boldsymbol{\mathcal{A}}}{\max}\;\left(R_G(s_\tau,\boldsymbol{a}_\tau,g)+\gamma\mathcal{P}^{\boldsymbol{a}}_{s's_\tau}v_G'(s')\right),\mathcal{M}^{\mathfrak{g},\boldsymbol{\pi}^c}Q_G(s_{\tau},\boldsymbol{a})\right\}\Bigg|
\\&\qquad\qquad+\Bigg|\max\left\{\underset{\boldsymbol{a}\in\boldsymbol{\mathcal{A}}}{\max}\;\left(R_G(s_\tau,\boldsymbol{a}_\tau,g)+\gamma\mathcal{P}^{\boldsymbol{a}}_{s's_\tau}v_G'(s')\right),\mathcal{M}^{\mathfrak{g},\boldsymbol{\pi}^c}Q_G(s_{\tau},\boldsymbol{a})\right\}-\underset{\boldsymbol{a}\in\boldsymbol{\mathcal{A}}}{\max}\;\left(R_G(s_\tau,\boldsymbol{a}_\tau,g)+\gamma\mathcal{P}^{\boldsymbol{a}}_{s's_\tau}v_G'(s')\right)\Bigg|
\\&\leq \gamma\underset{a\in\mathcal{A}}{\max}\;\left|\mathcal{P}^{\boldsymbol{\pi}}_{s's_\tau}\mathcal{P}^{\boldsymbol{a}}v_G(s')-\mathcal{P}^{\boldsymbol{\pi}}_{s's_\tau}\mathcal{P}^{\boldsymbol{a}}v_G'(s')\right|+\left|\max\left\{0,\mathcal{M}^{\mathfrak{g},\boldsymbol{\pi}^c}Q_G(s_{\tau},\boldsymbol{a})-\underset{\boldsymbol{a}\in\boldsymbol{\mathcal{A}}}{\max}\;\left(R_G(s_\tau,\boldsymbol{a}_\tau,g)+\gamma\mathcal{P}^{\boldsymbol{a}}_{s's_\tau}v_G'(s')\right)\right\}\right|
\\&\leq \gamma\left\|Pv_G-Pv_G'\right\|
\\&\leq \gamma\|v_G-v_G'\|,
\end{align*}
where we have used the fact that for any scalars $a,b,c$ we have that $
    \left|\max\{a,b\}-\max\{b,c\}\right|\leq \left|a-c\right|$ and the non-expansiveness of $P$.

\textbf{Case 2: }
Let us now consider the case:
\begin{align*}\mathcal{M}^{\mathfrak{g},\boldsymbol{\pi}^c}Q_G(s_{\tau},\boldsymbol{a})-\underset{\boldsymbol{a}\in\boldsymbol{\mathcal{A}}}{\max}\;\left(R_G(s_\tau,\boldsymbol{a}_\tau,g)+\gamma\mathcal{P}^{\boldsymbol{a}}_{s's_\tau}v_G'(s')\right)\geq 0.
\end{align*}

For this case, first recall that $c>0$, hence
\begin{align*}
&\mathcal{M}^{\mathfrak{g},\boldsymbol{\pi}^c}Q_G(s_{\tau},\boldsymbol{a})-\underset{\boldsymbol{a}\in\boldsymbol{\mathcal{A}}}{\max}\;\left(R_G(s_\tau,\boldsymbol{a}_\tau,g)+\gamma\mathcal{P}^{\boldsymbol{a}}_{s's_\tau}v_G'(s')\right)
\\&\leq \mathcal{M}^{\mathfrak{g},\boldsymbol{\pi}^c}Q_G(s_{\tau},\boldsymbol{a})-\underset{\boldsymbol{a}\in\boldsymbol{\mathcal{A}}}{\max}\;\left(R_G(s_\tau,\boldsymbol{a}_\tau,g)+\gamma\mathcal{P}^{\boldsymbol{a}}_{s's_\tau}v_G'(s')\right)+c
\\&\leq \left(R_G(s_\tau,\boldsymbol{a},g)- c+\gamma\mathcal{P}^{\boldsymbol{\pi}}_{s's_\tau}\mathcal{P}^{\boldsymbol{a}}v_G(s')\right)|^{\boldsymbol{a}\sim\boldsymbol{\pi}^c}-\underset{\boldsymbol{a}\in\boldsymbol{\mathcal{A}}}{\max}\;\left(R_G(s_\tau,\boldsymbol{a}_\tau,g)- c+\gamma\mathcal{P}^{\boldsymbol{a}}_{s's_\tau}v_G'(s')\right)
\\&\leq \underset{\boldsymbol{a}\in\boldsymbol{\mathcal{A}}}{\max}\;\left(R_G(s_\tau,\boldsymbol{a},g)- c+\gamma\mathcal{P}^{\boldsymbol{\pi}}_{s's_\tau}\mathcal{P}^{\boldsymbol{a}}v_G(s')\right)-\underset{\boldsymbol{a}\in\boldsymbol{\mathcal{A}}}{\max}\;\left(R_G(s_\tau,\boldsymbol{a}_\tau,g)- c+\gamma\mathcal{P}^{\boldsymbol{a}}_{s's_\tau}v_G'(s')\right)
\\&\leq \gamma\underset{\boldsymbol{a}\in\boldsymbol{\mathcal{A}}}{\max}\;\left|\mathcal{P}^{\boldsymbol{\pi}}_{s's_\tau}\mathcal{P}^{\boldsymbol{a}}\left(v_G(s')-v_G'(s')\right)\right|
\\&\leq \gamma\left|v_G(s')-v_G'(s')\right|
\\&\leq \gamma\left\|v_G-v_G'\right\|,
\end{align*}
 using the non-expansiveness of the operator $P$. Hence we have that
\begin{align}
    \left\|\mathcal{M}^{\mathfrak{g},\boldsymbol{\pi}^c}Q_G-\underset{\boldsymbol{a}\in\boldsymbol{\mathcal{A}}}{\max}\;\left[ R_G(\cdot,\boldsymbol{a})+\gamma\mathcal{P}^{\boldsymbol{a}}v_G'\right]\right\|\leq \gamma\left\|v_G-v_G'\right\|.\label{off_M_bound_gen}
\end{align}
Gathering the results of the three cases gives the desired result.

To prove the theorem, we make use of the following result:
\begin{theorem}[Theorem 1, pg 4 in \citep{jaakkola1994convergence}]
Let $\Xi_t(s)$ be a random process that takes values in $\mathbb{R}^n$ and given by the following:
\begin{align}
    \Xi_{t+1}(s)=\left(1-\alpha_t(s)\right)\Xi_{t}(s)\alpha_t(s)L_t(s),
\end{align}
then $\Xi_t(s)$ converges to $0$ with probability $1$ under the following conditions:
\begin{itemize}
\item[i)] $0\leq \alpha_t\leq 1, \sum_t\alpha_t=\infty$ and $\sum_t\alpha_t<\infty$
\item[ii)] $\|\mathbb{E}[L_t|\mathcal{F}_t]\|\leq \gamma \|\Xi_t\|$, with $\gamma <1$;
\item[iii)] ${\rm Var}\left[L_t|\mathcal{F}_t\right]\leq c(1+\|\Xi_t\|^2)$ for some $c>0$.
\end{itemize}
\end{theorem}
\begin{proof}
To prove the result, we show (i) - (iii) hold. Condition (i) holds by choice of learning rate. It therefore remains to prove (ii) - (iii). We first prove (ii). For this, we consider our variant of the Q-learning update rule:
\begin{align*}
Q_{t+1}(s_t,\boldsymbol{a}_t)=Q_{t}&(s_t,\boldsymbol{a}_t)
\\&+\alpha_t(s_t,\boldsymbol{a}_t)\left[\max\left\{\mathcal{M}^{\mathfrak{g},\boldsymbol{\pi}^c}Q(s_{\tau_k},\boldsymbol{a}), R(s_{\tau_k},\boldsymbol{a},g)+\gamma\underset{\boldsymbol{a'}\in\boldsymbol{\mathcal{A}}}{\max}\;Q_Gs_{t+1},\boldsymbol{a'})\right\}-Q_{t}(s_t,\boldsymbol{a}_t)\right].
\end{align*}
After subtracting $Q^\star(s_t,\boldsymbol{a}_t)$ from both sides and some manipulation we obtain that:
\begin{align*}
&\Xi_{t+1}(s_t,\boldsymbol{a}_t)
\\&=(1-\alpha_t(s_t,\boldsymbol{a}_t))\Xi_{t}(s_t,\boldsymbol{a}_t)
\\&\qquad\qquad\qquad\qquad\;\;+\alpha_t(s_t,\boldsymbol{a}_t))\left[\max\left\{\mathcal{M}^{\mathfrak{g},\boldsymbol{\pi}^c}Q_Gs_{\tau_k},\boldsymbol{a}), R_G(s_{\tau_k},\boldsymbol{a},g)+\gamma\underset{a'\in\mathcal{A}}{\max}\;Q_G(s',\boldsymbol{a'})\right\}-Q^\star(s_t,\boldsymbol{a}_t)\right],  \end{align*}
where $\Xi_{t}(s_t,\boldsymbol{a}_t):=Q_t(s_t,\boldsymbol{a}_t)-Q^\star(s_t,\boldsymbol{a}_t)$.

Let us now define by 
\begin{align*}
L_t(s_{\tau_k},\boldsymbol{a}):=\max\left\{\mathcal{M}^{\mathfrak{g},\boldsymbol{\pi}^c}Q_Gs_{\tau_k},\boldsymbol{a}), R_G(s_{\tau_k},\boldsymbol{a},g)+\gamma\underset{a'\in\mathcal{A}}{\max}\;Q_G(s',\boldsymbol{a'})\right\}-Q^\star(s_t,a).
\end{align*}
Then
\begin{align}
\Xi_{t+1}(s_t,\boldsymbol{a}_t)=(1-\alpha_t(s_t,\boldsymbol{a}_t))\Xi_{t}(s_t,\boldsymbol{a}_t)+\alpha_t(s_t,\boldsymbol{a}_t))\left[L_t(s_{\tau_k},a)\right].   
\end{align}

We now observe that
\begin{align}\nonumber
\mathbb{E}\left[L_t(s_{\tau_k},\boldsymbol{a})|\mathcal{F}_t\right]&=\sum_{s'\in\mathcal{S}}P(s';a,s_{\tau_k})\max\left\{\mathcal{M}^{\mathfrak{g},\boldsymbol{\pi}^c}Q_Gs_{\tau_k},\boldsymbol{a}), R_G(s_{\tau_k},\boldsymbol{a},g)+\gamma\underset{a'\in\mathcal{A}}{\max}\;Q_G(s',\boldsymbol{a'})\right\}-Q^\star(s_{\tau_k},a)
\\&= T_G Q_t(s,\boldsymbol{a})-Q^\star(s,\boldsymbol{a}). \label{expectation_L}
\end{align}
Now, using the fixed point property that implies $Q^\star=T_G Q^\star$, we find that
\begin{align}\nonumber
    \mathbb{E}\left[L_t(s_{\tau_k},\boldsymbol{a})|\mathcal{F}_t\right]&=T_G Q_t(s,\boldsymbol{a})-T_G Q^\star(s,\boldsymbol{a})
    \\&\leq\left\|T_G Q_t-T_G Q^\star\right\|\nonumber
    \\&\leq \gamma\left\| Q_t- Q^\star\right\|_\infty=\gamma\left\|\Xi_t\right\|_\infty.
\end{align}
using the contraction property of $T$ established in Lemma \ref{lemma:bellman_contraction}. This proves (ii).

We now prove iii), that is
\begin{align}
    {\rm Var}\left[L_t|\mathcal{F}_t\right]\leq c(1+\|\Xi_t\|^2).
\end{align}
Now by \eqref{expectation_L} we have that
\begin{align*}
  {\rm Var}\left[L_t|\mathcal{F}_t\right]&= {\rm Var}\left[\max\left\{\mathcal{M}^{\mathfrak{g},\boldsymbol{\pi}^c}Q_Gs_{\tau_k},\boldsymbol{a}), R_G(s_{\tau_k},\boldsymbol{a},g)+\gamma\underset{a'\in\mathcal{A}}{\max}\;Q_G(s',\boldsymbol{a'})\right\}-Q^\star(s_t,a)\right]
  \\&= \mathbb{E}\Bigg[\Bigg(\max\left\{\mathcal{M}^{\mathfrak{g},\boldsymbol{\pi}^c}Q_Gs_{\tau_k},\boldsymbol{a}), R_G(s_{\tau_k},\boldsymbol{a},g)+\gamma\underset{a'\in\mathcal{A}}{\max}\;Q_G(s',\boldsymbol{a'})\right\}
  \\&\qquad\qquad\qquad\qquad\qquad\quad\quad\quad-Q^\star(s_t,a)-\left(T_G Q_t(s,\boldsymbol{a})-Q^\star(s,\boldsymbol{a})\right)\Bigg)^2\Bigg]
      \\&= \mathbb{E}\left[\left(\max\left\{\mathcal{M}^{\mathfrak{g},\boldsymbol{\pi}^c}Q_Gs_{\tau_k},\boldsymbol{a}), R_G(s_{\tau_k},\boldsymbol{a},g)+\gamma\underset{a'\in\mathcal{A}}{\max}\;Q_G(s',\boldsymbol{a'})\right\}-T_G Q_t(s,\boldsymbol{a})\right)^2\right]
    \\&= {\rm Var}\left[\max\left\{\mathcal{M}^{\mathfrak{g},\boldsymbol{\pi}^c}Q_Gs_{\tau_k},\boldsymbol{a}), R_G(s_{\tau_k},\boldsymbol{a},g)+\gamma\underset{a'\in\mathcal{A}}{\max}\;Q_G(s',\boldsymbol{a'})\right\}-T_G Q_t(s,\boldsymbol{a}))^2\right]
    \\&\leq c(1+\|\Xi_t\|^2),
\end{align*}
for some $c>0$ where the last line follows due to the boundedness of $Q$ (which follows from Assumptions 2 and 4). This concludes the proof of part (i) of the Theorem (i.e. \textbf{[A]}).

\end{proof}
\end{proof}

\subsection*{Proof of Part \textbf{B}}

To prove Part \textbf{B}, we prove the following result\footnote{This property is analogous to  the condition in Markov potential games \citep{macua2018learning,mguni2021learning}} :
\begin{proposition}\label{dpg_proposition}
For any ${\pi}\in{\Pi}$ and for any {\fontfamily{cmss}\selectfont Global} policy $\mathfrak{g}$, there exists a function $B^{\boldsymbol{\pi},\mathfrak{g}}:\mathcal{S}\times \{0,1\}\to \mathbb{R}$ such that
\begin{align}
 v(s|\boldsymbol{\pi})-v(s|\boldsymbol{\pi}')
&=B(s|\boldsymbol{\pi},\mathfrak{g})-B(s|\boldsymbol{\pi},\mathfrak{g}'),\;\;\forall s \in\cS
\\v_G(s|\boldsymbol{\pi},\mathfrak{g})-v_G(s|\boldsymbol{\pi},\mathfrak{g}')
&=B(s|\boldsymbol{\pi},\mathfrak{g})-B(s|\boldsymbol{\pi},\mathfrak{g}''),\;\;\forall s \in\cS
\\ v_G(s|\boldsymbol{\pi},\mathfrak{g})-v_G(s|\boldsymbol{\pi}',\mathfrak{g})
&=B(s|\boldsymbol{\pi},\mathfrak{g})-B(s|\boldsymbol{\pi},\mathfrak{g}'),\;\;\forall s \in\cS
\label{potential_relation_proof}
\end{align}
where in particular the function $B$ is given by:
\begin{align}
B(s|\boldsymbol{\pi},\mathfrak{g}) =\mathbb{E}\left[\sum_{t=0}^\infty \gamma^tr\right],\end{align}
for any $s\in\mathcal{S}$.
\end{proposition}
\begin{proof}
This is manifest from the construction of $B$ and Assumption 6.
\end{proof}

\section*{Proof of Proposition \ref{NE_improve_prop}}
\begin{proof}[Proof of Prop. \ref{NE_improve_prop}]
We split the proof into two parts: 

i) We first prove that $v^{\boldsymbol{\tilde{\pi}}}(s)\geq v^{\boldsymbol{\pi}}(s),\;\forall s \in\mathcal{S}$ where we use $\boldsymbol{\tilde{\pi}}$ to denote the $N$ agents' joint policy induced under the influence of the {\fontfamily{cmss}\selectfont Global}.

ii) Second, we prove that there exists a finite integer $M$ such that  $v^{\boldsymbol{\tilde{\pi}}_m}(s)\geq v^{\boldsymbol{\pi}_m}(s)$ for any $m\geq M$.

The proof of part (i) is achieved by proof by contradiction.
Denote by $v^{\boldsymbol{\pi},\mathfrak{g}\equiv {0}}$ the value function for the {\fontfamily{cmss}\selectfont Controller} for the system \textit{without the {\fontfamily{cmss}\selectfont Global}}. 
Indeed, 
let $(\hat{\boldsymbol{\pi}},\mathfrak{g})$ be the policy profile at the stable point of the system (Markov perfect equilibrium) and assume that {\fontfamily{cmss}\selectfont Global}'s interventions lead to a decrease in total system returns. Then by construction 
$    v^{\hat{\boldsymbol{\pi}},\mathfrak{g}}(s)< v^{\boldsymbol{\pi},\mathfrak{g}\equiv {0}}(s)
$
which is a contradiction since $(\hat{\boldsymbol{\pi}},\mathfrak{g})$ is a stable point (MPE profile). 

We now prove part (ii). 

By part (i) we have that $v^{\boldsymbol{\tilde{\pi}}}(s)=\underset{m\to\infty}{\lim}v^{\boldsymbol{\tilde{\pi}}_m}(s)\geq v^{\boldsymbol{\pi}}(s)= \underset{m\to\infty}{\lim}v^{\boldsymbol{\pi}_m}(s)$. Since $v^{\boldsymbol{\pi}}(s)$ is maximal in the sequence $v^{\boldsymbol{\pi}_1}(s),v^{\boldsymbol{\pi}_2}(s),\ldots,v^{\boldsymbol{\pi}}(s)$ we can deduce that $v^{\boldsymbol{\pi}}(s)\geq v^{\boldsymbol{\pi}_n}(s)$ for any $n\leq \infty$. Hence for any $n$ there exists a $c\geq 0$ such that $\underset{m\to\infty}{\lim}v^{\boldsymbol{\tilde{\pi}}_m}(s)= v^{\boldsymbol{\pi}_n}(s)- c$.
Now by construction $v^{\boldsymbol{\tilde{\pi}}_m}(s)\to v^{\boldsymbol{\tilde{\pi}}}(s)$ as $m\to\infty$, therefore the sequence $\tilde{v}^{\boldsymbol{\pi}_1},\tilde{v}^{\boldsymbol{\pi}_2},\ldots,$ forms a Cauchy sequence. Therefore, there exists an $M$ such that for any $\epsilon >0$, $v^{\tilde{\boldsymbol{\pi}}_n}(s)- (v^{\boldsymbol{\pi}_n}(s)- c)< \epsilon$ $\forall n\geq M$. Since $\epsilon$ is arbitrary we can conclude that $v^{\tilde{\boldsymbol{\pi}}_n}(s)- (v^{\boldsymbol{\pi}_n}(s)- c)= 0$ $\forall n\geq M$. Since $c\geq 0$, we immediately deduce that $v^{\tilde{\boldsymbol{\pi}}_n}(s)\geq v^{\boldsymbol{\pi}_n}(s), \forall n\geq M$  which is the required result.    
\end{proof}

\section*{Proof of Proposition \ref{prop:switching_times}}
\begin{proof}
We begin by re-expressing the \textit{activation times} at which the {\fontfamily{cmss}\selectfont Global} agent activates \textbf{{\fontfamily{cmss}\selectfont Central}}. In particular,an activation time $\tau_k$ is defined recursively $\tau_k=\inf\{t>\tau_{k-1}|s_t\in A,\tau_k\in\mathcal{F}_t\}$ where $A=\{s\in \mathcal{S},g(s_t)=1\}$.
The proof is given by deriving a contradiction.  Let us there suppose that $\mathcal{M}v_G(s_{\tau_k})\leq v_G(s_{\tau_k})$ and that the activation time $\tau'_1>\tau_1$ is an optimal activation time. Construct the $\mathfrak{g}'$ and $\mathfrak{g}$ policy switching times by $(\tau'_0,\tau'_1,\ldots,)$ and $(\tau'_0,\tau_1,\ldots)$ respectively.  Define by $l=\inf\{t>0;\mathcal{M}v_G(s_t)= v_G(s_t)\}$ and $m=\sup\{t;t<\tau'_1\}$.
By construction we have that
{\small
\begin{align*}
& \quad v_G(s|\boldsymbol{\pi},\mathfrak{g}')
\\&=\mathbb{E}\left[R_G(s_{0},\boldsymbol{a}_{0},g)+\mathbb{E}\left[\ldots+\gamma^{l-1}\mathbb{E}\left[R(s_{\tau_1-1},\boldsymbol{a}_{\tau_1-1},g)+\ldots+\gamma^{m-l-1}\mathbb{E}\left[ R_G(s_{\tau'_1-1},\boldsymbol{a}_{\tau'_1-1},g)+\gamma\mathcal{M}^{\boldsymbol{\pi},\mathfrak{g}'}v_G(s_{\tau_1}|\boldsymbol{\pi},\mathfrak{g}')\right]\right]\right]\right]
\\&<\mathbb{E}\left[R_G(s_{0},\boldsymbol{a}_{0},g)+\mathbb{E}\left[\ldots+\gamma^{l-1}\mathbb{E}\left[ R_G(s_{\tau_1-1},\boldsymbol{a}_{\tau_1-1},g)+\gamma\mathcal{M}^{\boldsymbol{\pi},\tilde{\mathfrak{g}}}v_G(s_{\tau_1}|\boldsymbol{\pi},\mathfrak{g}')\right]\right]\right]
\end{align*}}
We make use of the following observation 
\begin{align}
&\mathbb{E}\left[ R_G(s_{\tau_1-1},\boldsymbol{a}_{\tau_1-1},g)+\gamma\mathcal{M}^{\boldsymbol{\pi},\tilde{\mathfrak{g}}}v_G(s_{\tau_1}|\boldsymbol{\pi},\mathfrak{g}')\right]
\\&\leq \max\left\{\mathcal{M}^{\boldsymbol{\pi},\tilde{\mathfrak{g}}}v_G(s_{\tau_1}|\boldsymbol{\pi},\mathfrak{g}'),\underset{\boldsymbol{a}_{\tau_1}\in\boldsymbol{\mathcal{A}}}{\max}\;\left[ R_G(s_{\tau_1},\boldsymbol{a}_{\tau_1},g)+\gamma\sum_{s'\in\mathcal{S}}P(s';a_{\tau_1},s_{\tau_1})v_G(s'|\boldsymbol{\pi},\mathfrak{g})\right]\right\}.
\end{align}

Using this we deduce that
{\tiny
\begin{align*}
&v_G(s|\boldsymbol{\pi},\mathfrak{g}')\leq\mathbb{E}\Bigg[R_G(s_{0},\boldsymbol{a}_{0},g)+\mathbb{E}\Bigg[\ldots
\\&+\gamma^{l-1}\mathbb{E}\left[ R_G(s_{\tau_1-1},\boldsymbol{a}_{\tau_1-1},g)+\gamma\max\left\{\mathcal{M}^{\boldsymbol{\pi},\tilde{\mathfrak{g}}}v_G(s_{\tau_1}|\boldsymbol{\pi},\mathfrak{g}'),\underset{\boldsymbol{a}_{\tau_1}\in\boldsymbol{\mathcal{A}}}{\max}\;\left[ R_G(s_{\tau_{k}},\boldsymbol{a}_{\tau_{k}},g)+\gamma\sum_{s'\in\mathcal{S}}P(s';a_{\tau_1},s_{\tau_1})v_G(s'|\boldsymbol{\pi},\mathfrak{g}),\right]\right\}\right]\Bigg]\Bigg]
\\&=\mathbb{E}\left[R_G(s_{0},\boldsymbol{a}_{0},g)+\mathbb{E}\left[\ldots+\gamma^{l-1}\mathbb{E}\left[ R_G(s_{\tau_1-1},\boldsymbol{a}_{\tau_1-1},g)+\gamma\left[T_G v_G(s_{\tau_1}|\boldsymbol{\pi},\tilde{\mathfrak{g}})\right]\right]\right]\right]=v_G(s|\boldsymbol{\pi},\tilde{\mathfrak{g}}),
\end{align*}}
where the first inequality is true by assumption on $\mathcal{M}$. This is a contradiction since $\pi'$ is an optimal policy for Player 2. Using analogous reasoning, we deduce the same result for $\tau'_k<\tau_k$ after which deduce the result. Moreover, by invoking the same reasoning, we can conclude that it must be the case that $(\tau_0,\tau_1,\ldots,\tau_{k-1},\tau_k,\tau_{k+1},\ldots,)$ are the optimal switching times. This completes the proof.
\end{proof}
\section{Proof of Theorem \ref{thm:optimal_policy_budget}}
\begin{proof}
The proof of the Theorem is straightforward since by Theorem \ref{theorem:existence}, {\fontfamily{cmss}\selectfont Global}'s problem can be solved using a dynamic programming principle. The proof immediately by application of Theorem 2 in \cite{sootla2022saute}.

\end{proof}

\end{document}